\documentclass[11pt]{article}

\usepackage[utf8]{inputenc}
\usepackage[margin=1in]{geometry}
\usepackage[normalem]{ulem}
\usepackage{mathtools,bm}
\usepackage[inline]{enumitem}
\usepackage{amsmath,amsfonts,amssymb,amsthm,hyperref,algorithm,eqparbox,array,ragged2e,stmaryrd,bbm}
\usepackage{xspace}

\usepackage{tocloft}
\usepackage[nottoc]{tocbibind} 

\usepackage[usenames,dvipsnames,svgnames,table]{xcolor}
\usepackage{tikz}

\usetikzlibrary{
    babel,
    arrows,
	calc,
	patterns,
	intersections,
    backgrounds,
    trees,
    mindmap,
    fadings,
	decorations.pathreplacing,
	decorations.pathmorphing,
	decorations.text,
	decorations.markings,
    scopes,
	shapes,
    snakes,
    shapes,
	positioning,
    arrows.meta,
    patterns,
    patterns.meta,
    angles,
    quotes,
    math
}

\tikzfading[name = middle,
    top color = transparent!100,
    bottom color = transparent!100,
    middle color = transparent!30
]

\newcommand{\addconnectors}[5]{
    \ifthenelse{#2 = 0}{\node[graph-vert] (#11) at (#19) {};}{\node[active-vert] (#11) at (#19) {};}
    \ifthenelse{#3 = 0}{\node[graph-vert] (#12) at (#110) {};}{\node[active-vert] (#12) at (#110) {};}
    \ifthenelse{#4 = 0}{\node[graph-vert] (#13) at (#111) {};}{\node[active-vert] (#13) at (#111) {};}
    \ifthenelse{#5 = 0}{\node[graph-vert] (#14) at (#112) {};}{\node[active-vert] (#14) at (#112) {};}
}

\tikzset{
    >=latex,
    pics/rhom/.style n args = {4}{
        code = {
            \draw[
                thick, #3, fill = #3!50, fill opacity = 0.2, rounded corners = 3pt]
                (-#1, 0) -- (0, -#2) -- (#1, 0) -- (0, #2) -- cycle;
            \node[draw, fill, circle, inner sep = 0pt, minimum size = 0.05cm, #3] (#4) at (0, 0) {};
        }
    },
    linecol/.style n args = {3}{
        postaction = {
            decorate,
            decoration = {
                markings,
                mark = between positions 0 and 0.9 step 0.2pt with {
                    \pgfmathsetmacro\myval{
                        multiply(
                            divide(
                                \pgfkeysvalueof{/pgf/decoration/mark info/distance from start},
                                \pgfdecoratedpathlength
                            ),
                            120
                        )};
                    \pgfsetfillcolor{#3!\myval!#2};
                    \pgfpathcircle{\pgfpointorigin}{#1};
                    \pgfusepath{fill};
                },
                mark = at position 1 with {\arrow[#3]{latex}},
            }
        }
    },
    pics/universe-rect/.style n args = {2}{
        code = {
            \fill[black!10, rounded corners = 2pt] (-0.07, -0.07) rectangle (#1 + 0.07, #2 + 0.07);
            \draw[fill = white, thick] (0, 0) rectangle (#1, #2);
        }
    },
    pics/comm-rect/.style n args = {4}{
        code = {
            \draw[
                thick, #3, fill = #3!50, fill opacity = 0.2, rounded corners = 3pt]
                (-#1, -#2) rectangle (#1, #2);
            \node[draw, fill, circle, inner sep = 0pt, minimum size = 0.05cm, #3] (#4) at (0, 0) {};
        }
    },
    graph-vert/.style 2 args = {
        draw,
        color = #1!70!black,
        circle,
        thick,
        inner sep = 0pt,
        minimum size = #2,
        fill = #1!50,
        fill opacity = 0.5
    },
    non-path/.style = {
        thick,
        black
    },
    active-path/.style = {
        VioletRed,
        thick,
        decorate,
        decoration = {snake, segment length = 1mm, amplitude = 0.1mm}
    },
    graph-vert/.default = {black}{0.1cm},
    active-vert/.style = {
        graph-vert = {VioletRed}{0.1cm},
    },
    pics/swap/.style n args = {3}{
        code = {
            \fill[black!10, rounded corners = 2pt] (-0.05, -0.05) rectangle (1.05, 0.55);
            \foreach \p [count = \ii] in {(0.2, 0), (0.8, 0)}{
                \pgfmathtruncatemacro{\ff}{\ii + 4}
                \pgfmathtruncatemacro{\t}{\ii + 8}
                \path \p coordinate (#1\ff) -- ++(0, -0.2) coordinate (#1\t);
            }
            \foreach \p [count = \ii] in {(0.2, 0.5), (0.8, 0.5)}{
                \pgfmathtruncatemacro{\ff}{\ii + 6}
                \pgfmathtruncatemacro{\t}{\ii + 10}
                \path \p coordinate (#1\ff) -- ++(0, 0.2) coordinate (#1\t);
            }
            \ifthenelse{#3 = 0}{
                \foreach \ii in {1, 2, 3, 4}{
                    \pgfmathtruncatemacro{\ff}{\ii + 4}
                    \pgfmathtruncatemacro{\t}{\ii + 8}

                    \addconnectors{#1}{0}{0}{0}{0}
                    \node[graph-vert] (#1\ii) at (#1\t) {};
                    \draw[non-path] (#1\ii) -- (#1\ff);
                }

                \fill[white] (0, 0) rectangle (1, 0.5);
            }{
                \fill[ForestGreen!10] (0, 0) rectangle (1, 0.5);
            }
            \ifthenelse{\tmp = 0}{
                \node at (0.5,0.25) {?};
            }{}
            \ifthenelse{\tmp = 1}{
                \addconnectors{#1}{1}{0}{1}{0}
                \draw[active-path] (#11) -- (#13);
                \draw[non-path] (#12) -- (#14);
            }{}
            \ifthenelse{\tmp = 2}{
                \addconnectors{#1}{0}{1}{0}{1}
                \draw[non-path] (#11) -- (#13);
                \draw[active-path] (#12) -- (#14);
            }{}
            \ifthenelse{\tmp = 3}{
                \addconnectors{#1}{1}{0}{0}{1}
                \draw[active-path]
                    (#11) -- (#15) .. controls ++(0, 0.2) and ++(0, -0.2) .. (#18) -- (#14);
                \draw[non-path]
                    (#12) -- (#16) .. controls ++(0, 0.2) and ++(0, -0.2) .. (#17) -- (#13);
            }{}
            \ifthenelse{\tmp = 4}{
                \addconnectors{#1}{0}{1}{1}{0}
                \draw[non-path]
                    (#11) -- (#15) .. controls ++(0, 0.2) and ++(0, -0.2) .. (#18) -- (#14);
                \draw[active-path]
                    (#12) -- (#16) .. controls ++(0, 0.2) and ++(0, -0.2) .. (#17) -- (#13);
            }{}
            \ifthenelse{\tmp = 5}{
                \addconnectors{#1}{1}{1}{1}{1}
                \draw[active-path] (#11) -- (#13);
                \draw[active-path] (#12) -- (#14);
            }{}
            \ifthenelse{\tmp = 6}{
                \addconnectors{#1}{1}{1}{1}{1}
                \draw[active-path]
                    (#11) -- (#15) .. controls ++(0, 0.2) and ++(0, -0.2) .. (#18) -- (#14);
                \draw[active-path]
                    (#12) -- (#16) .. controls ++(0, 0.2) and ++(0, -0.2) .. (#17) -- (#13);
            }{}
            \draw[thick] (0, 0) rectangle (1, 0.5);
            \node[above] at (0.5, 0.5) {\scriptsize #2};
        }
    },
    fancy-arrow/.style = {
        draw = black,
        thick,
        single arrow,
        right color = VioletRed!50,
        left color = ForestGreen!50,
        fill opacity = 0.7,
        single arrow head extend = 0.3cm,
        single arrow tip angle = 70,
        single arrow head indent = 0.15cm,
        minimum height = 2.4cm,
        minimum width = 0.8cm,
        shading angle = -45,
        rotate = -90
    },
    in-out-node/.style = {
        draw,
        circle,
        thick,
        inner sep = 0,
        minimum size = 0.09cm,
        #1,
        fill = #1!50
    }
}

\usepackage{ifthen}
\usepackage{subfig}

\usepackage{thm-restate}
\usepackage{enumitem}
\usepackage[capitalise,nameinlink,noabbrev]{cleveref}
\usepackage[font=small,labelfont=bf]{caption} 

\hypersetup{
	colorlinks=true,
	urlcolor=MidnightBlue,
	linkcolor=MidnightBlue,
	citecolor=MidnightBlue,
	unicode
}
\usepackage{algorithm}
\usepackage{algpseudocode}

\algnewcommand\algorithmicinput{\textbf{Input: }}
\algnewcommand\Input{\item[\algorithmicinput]}
\algnewcommand\algorithmicoutput{\textbf{Output: }}
\algnewcommand\Output{\item[\algorithmicoutput]}
\algnewcommand{\OneLineIf}[2]{
  \State \algorithmicif\ #1\ \algorithmicthen\ #2}
  \DeclareSymbolFont{extraup}{U}{zavm}{m}{n}
  \DeclareMathSymbol{\vardiamond}{\mathalpha}{extraup}{87}

\newtheorem{theorem}{Theorem}
\newtheorem{theorem*}{Theorem}
\newtheorem{lemma}[theorem]{Lemma}

\newtheorem{claim}[theorem]{Claim}
\newtheorem{corollary}[theorem]{Corollary}

\newtheorem{fact}{Fact}
\theoremstyle{remark}

\theoremstyle{definition}

\newtheorem{question}{Question}

\makeatletter
\def\th@example{%
  \thm@notefont{}
  \normalfont 
}
\def\th@definition{%
  \thm@notefont{}
  \normalfont 
}
\makeatother

\theoremstyle{example}

\renewcommand{\setminus}{\smallsetminus}

\renewcommand{\H}{\operatorname{H}}

\newcommand{\tO}{\tilde{O}}
\newcommand{\tOmega}{\tilde{\Omega}}

\DeclareMathOperator*{\Exp}{\mathbb{E}}

\newcommand{\x}{\bm{x}}

\renewcommand{\k}{\bm{k}}
\renewcommand{\epsilon}{\varepsilon}

\newcommand{\Xor}{\textsc{Xor}} 

\newcommand{\lb}{\llbracket}
\newcommand{\rb}{\rrbracket}

\newcommand{\poly}{\mathrm{poly}}

\newcommand{\N}{\mathbb{N}}
\newcommand{\dist}{\mathrm{dist}}

\newcommand{\newmeasure}[2]{\newcommand{#1}{{\textup{\sffamily #2}}\xspace}}
\newmeasure{\C}{C}

\newcommand{\ACzero}{\textsf{\upshape AC}^0}

\newcommand{\supp}{\mathrm{supp}}

\let\oldprod=\prod 

\RenewDocumentCommand{\prod}{e{_^}}{ 
  \vphantom{\oldprod_{.}} 
  \mathop{\smash{\oldprod 
    \IfValueT{#1}{_{#1}} 
    \IfValueT{#2}{^{#2}}
  }}
}

\begin{document}

\begin{center}
\mbox{}\\[2cm]
    {\huge Sampling Permutations with Cell Probes is Hard}
    \\[1cm]
    \large

    \setlength\tabcolsep{1.2em}
    \begin{tabular}{ccc}
      Yaroslav Alekseev &
                          Mika G\"o\"os &
                                          Konstantin Myasnikov \\[-1mm]
      \small\slshape Technion &
                                \small\slshape EPFL&
                                                     \small\slshape EPFL
    \end{tabular}  \\[1mm]
    \begin{tabular}{cc}
      Artur Riazanov &
                       Dmitry Sokolov\\[-1mm]
      \small\slshape EPFL&
                           \small\slshape EPFL \& Universit{\'{e}} de Montr{\'{e}}al
    \end{tabular}
    
\vspace{2em}
\today
\vspace{1em}
\end{center}

\begin{abstract}\noindent
Suppose we are given an infinite sequence of input cells, each initialized with a uniform random symbol from $[n]$. How hard is it to output a sequence in $[n]^n$ that is close to a uniform random permutation? Viola ({\footnotesize SICOMP 2020}) conjectured that if each output cell is computed by making~$d$ probes to input cells, then $d\geq\omega(1)$. Our main result shows that, in fact,~$d\geq (\log n)^{\Omega(1)}$, which is tight up to the constant in the exponent. Our techniques also show that if the probes are \emph{nonadaptive}, then $d\geq n^{\Omega(1)}$, which is an exponential improvement over the previous nonadaptive lower bound due to Yu and Zhan ({\footnotesize ITCS 2024}). Our results also imply lower bounds against succinct data structures for storing permutations.
\end{abstract}

{
\vspace{15mm}

\setlength{\cftbeforesecskip}{0pt}
\renewcommand\cftsecfont{\mdseries}
\renewcommand{\cftsecpagefont}{\normalfont}
\renewcommand{\cftsecleader}{\cftdotfill{\cftdotsep}}
\setcounter{tocdepth}{1}
\tableofcontents
}
\thispagestyle{empty}
\setcounter{page}{0}

\newpage
\setcounter{page}{1}

\section{Introduction}
Randomly shuffling the elements of an array is one of the most basic primitives in randomized algorithms. It is a simple programming exercise to implement this in linear time~\cite{Dur64}. Doing it much faster, with a parallel algorithm, has been studied extensively~\cite{reif1985, MV91, CK00, Czumaj15}. In particular, array shuffling is possible even in constant-time in a parallel RAM model~\cite{Hag91}.

An analogous problem in probability theory is the question of \emph{card shuffling}, which dates back to Markov~\cite{Markov1906}. For example, one of the long-standing challenges in card shuffling has been to determine how many \emph{Thorp shuffles} (explained in~\cref{fig:thorp} below) are sufficient to produce a nearly uniform permutation over $n$ cards. A sequence of works~\cite{Morris05,MT06,Morris09,Morris13} has culminated in a result showing that $O(\log^3 n)$ shuffles are enough. An interesting feature of this shuffle (which has found applications, e.g., in cryptography~\cite{MRS09}) is its \emph{obliviousness}: the final position of each card can be computed by \emph{accessing only a few bits of randomness} (formally, the shuffle is given by a shallow ``switching network'').

\paragraph{Cell-probe model.}
A widely-studied computational model that captures oblivious shuffling (and much more) is the \emph{cell-probe model}~\cite{Yao81}. In this model, we are given a sequence of $s$ input cells, each storing a symbol from $[n]$. A cell-probe algorithm then produces an output sequence in $[n]^m$ where each output cell is computed by making $d$ probes (queries) to input cells. That is, the algorithm computes a function $f\colon [n]^{s} \to [n]^m$, where the $i$-th output cell $f_i$ is computed by a depth-$d$ arity-$n$ decision tree. Such cell-probe algorithms are also called depth-$d$ \emph{decision forests}.

\begin{figure}[hbt]
\vspace{1em}
    \centering
    \subfloat{\begin{tikzpicture}[scale=0.9]
    \def\k{4}
    \def\swaparray{{{0, 0, 2, 0}, {0, 3, 0, 0}, {0, 0, 1, 0}}}
    \def\inputarray{{
        "0", "0", "0", "0", "0", "0", "1", "0", "0", "0",
        "0", "0", "0", "0", "0", "0"
    }}
    \def\pfs{{
        "100", "100", "100"
    }}
    \def\pf{{
        "100", "100", "100"
    }}
    \def\pss{{
        "3", "2", "3"
    }}
    \def\ps{{
        "2", "1", "1"
    }}

    \node at (0, 2.5) {};
    \node at (0, -2.2) {};

    \foreach \j in {1, 2, 3}{
        \foreach \i in {1, 2, ..., \k}{
            \pgfmathtruncatemacro{\num}{(\j - 1) * \k + \i}
            \pgfmathtruncatemacro{\tmp}{\swaparray[\j - 1][\i - 1]}
            \pic[shift = {({1.7 * (\i - 1)}, {-1.8 * (\j - 1)})}]
                {swap = {a\j\i}{$r_{\num}$}{\tmp}};
        }
    }

    \foreach \j in {1, 2}{
        \pgfmathtruncatemacro{\t}{\j + 1}
        
        \foreach \i in {1, 2, ..., \k}{
            \pgfmathtruncatemacro{\swap}{(\i - 1) / 2 + 1}
            \pgfmathtruncatemacro{\swapn}{(\i - 1) / 2 + 1 + \k / 2}
            \pgfmathtruncatemacro{\g}{mod(\i - 1, 2) + 1}
            \pgfmathtruncatemacro{\gg}{\g + 8}

            \pgfmathparse{\pfs[\j - 1]}
            \edef\pfirswap{\pgfmathresult}
            \pgfmathparse{\pf[\j - 1]}
            \edef\pfirout{\pgfmathresult}

            \pgfmathparse{\pss[\j - 1]}
            \edef\psecswap{\pgfmathresult}
            \pgfmathparse{\ps[\j - 1]}
            \edef\psecout{\pgfmathresult}

            \pgfmathtruncatemacro{\flag}{0}

            \ifthenelse{\swap = \pfirswap \AND \pfirout = \g}{
                \pgfmathtruncatemacro{\flag}{1}
                \draw[active-path]
                    (a\j\swap\g) .. controls ++(0, -0.5) and ++(0, 0.5) .. (a\t\i3);
            }{}

            \ifthenelse{\swap = \psecswap \AND \psecout = \g}{
                \pgfmathtruncatemacro{\flag}{1}
                \draw[active-path]
                    (a\j\swap\g) .. controls ++(0, -0.5) and ++(0, 0.5) .. (a\t\i3);
            }{}

            \ifthenelse{\flag = 0}{
                \draw[non-path]
                    (a\j\swap\g) .. controls ++(0, -0.5) and ++(0, 0.5) .. (a\t\i3);

            }{}

            \pgfmathtruncatemacro{\flag}{0}

            \ifthenelse{\swapn = \pfirswap \AND \pfirout = \g}{
                \pgfmathtruncatemacro{\flag}{1}
                \draw[active-path]
                    (a\j\swapn\g) .. controls ++(0, -0.5) and ++(0, 0.5) .. (a\t\i4);
            }{}

            \ifthenelse{\swapn = \psecswap \AND \psecout = \g}{
                \pgfmathtruncatemacro{\flag}{1}
                \draw[active-path]
                    (a\j\swapn\g) .. controls ++(0, -0.5) and ++(0, 0.5) .. (a\t\i4);
            }{}

            \ifthenelse{\flag = 0}{
                \draw[non-path]
                    (a\j\swapn\g) .. controls ++(0, -0.5) and ++(0, 0.5) .. (a\t\i4);
            }{}
        }
    }

    \foreach \swap in {1, 2, ..., \k}{
        \foreach \connect in {3, 4}{
            \pgfmathtruncatemacro{\ind}{2 * (\swap - 1) - 2 + \connect}
            \pgfmathtruncatemacro{\tmp}{\inputarray[\ind - 1]}

            \path (a1\swap\connect) -- ++(0, 1.2) coordinate (topcon\ind);
            \fill[black!10, rounded corners = 2pt]
                (topcon\ind) ++(-0.25, -0.05) rectangle ++(0.5, 0.5);

            \ifthenelse{\tmp = 0}{
                \draw[thick, fill = white] (topcon\ind) ++(-0.2, 0) rectangle ++(0.4, 0.4)
                    node[midway, non-path] {$\ind$};
            }{
                \draw[VioletRed, thick, fill = ForestGreen!10] (topcon\ind) ++(-0.2, 0) rectangle ++(0.4, 0.4)
                    node[midway] {$\ind$};
                    
            }
        }
    }

    \foreach \i in {1, 2, ..., \k}{
        \pgfmathtruncatemacro{\tmp}{\i + \k}
        \pgfmathtruncatemacro{\swapn}{(\i - 1) / 2 + 1 + \k / 2}
        \pgfmathtruncatemacro{\g}{mod(\i - 1, 2) + 1}

        \pgfmathparse{\inputarray[\i - 1]}
        \edef\check{\pgfmathresult}

        \ifthenelse{\check = 0}{
            \draw[non-path]
                (topcon\i) .. controls ++(0, -0.5) and ++(0, 0.5) .. (a1\i3);
        }{
            \draw[active-path]
                (topcon\i) .. controls ++(0, -0.5) and ++(0, 0.5) .. (a1\i3);
        }

        \pgfmathparse{\inputarray[\tmp - 1]}
        \edef\check{\pgfmathresult}

        \ifthenelse{\check = 0}{
            \draw[non-path]
                (topcon\tmp) .. controls ++(0, -0.5) and ++(0, 0.5) .. (a1\i4);
        }{
            \draw[active-path]
                (topcon\tmp) .. controls ++(0, -0.5) and ++(0, 0.5) .. (a1\i4);
        }

    }    

    \node[VioletRed, below = 3pt] at (a331) {$7$};

\end{tikzpicture}}
    \hspace{1cm}
    \subfloat{\begin{tikzpicture}[scale=1]
    \tikzstyle{size-tree} = [
        inner sep = 0,
        minimum size = 0.5cm]
    \tikzstyle{active-tree} = [
        VioletRed,
        thick,
        draw,
        circle,
        fill = ForestGreen!10,
        size-tree]
    \tikzstyle{non-path-tree} = [
        thick,
        circle,
        draw,
        size-tree]

    \def\h{1.89}
    \def\sh{1.25}
    \def\inputarray{{
        "5", "1", "7", "3", "6", "2", "8", "4"
    }}

    \node at (0, 1.5) {};
    \node at (0, -5) {};

    \node[active-tree] (a) at (0, \sh) {\footnotesize $r_{11}$};
    \node[active-tree] (a1) at (-2, \sh - \h) {\footnotesize $r_{6}$};
    \node[non-path-tree] (a2) at (2, \sh - \h) {\footnotesize $r_{8}$};
    
    \node[non-path-tree] (a11) at (-3, \sh - 2 * \h) {\footnotesize $r_{1}$};
    \node[active-tree] (a12) at (-1, \sh - 2 * \h) {\footnotesize $r_{3}$};
    \node[non-path-tree] (a21) at (1, \sh - 2 * \h) {\footnotesize $r_{2}$};
    \node[non-path-tree] (a22) at (3, \sh - 2 * \h) {\footnotesize $r_{4}$};

    \draw[->, active-path] (a) -- (a1);
    \draw[->, non-path] (a) -- (a2);
    \draw[->, non-path] (a1) -- (a11);
    \draw[->, active-path] (a1) -- (a12);
    \draw[->, non-path] (a2) -- (a21);
    \draw[->, non-path] (a2) -- (a22);

    \foreach \xpos [count = \i] in {-3.4, -2.6, -1.4, -0.6, 0.6, 1.4, 2.6, 3.4}{
        \pgfmathtruncatemacro{\tmp}{\inputarray[\i - 1]}
        \coordinate (b) at (\xpos, \sh - 3 * \h + 0.02) {};
        \coordinate (b\i) at (\xpos, \sh - 3 * \h + 0.02 + 0.4) {};
        \fill[black!10, rounded corners = 2pt]
            (b) ++(-0.25, -0.05) rectangle ++(0.5, 0.5);

        \ifthenelse{\i = 3}{
            \draw[VioletRed, thick, fill = ForestGreen!10] (b) ++(-0.2, 0) rectangle ++(0.4, 0.4)
                node[midway] {$\tmp$};
        }{
            \draw[thick, fill = white] (b) ++(-0.2, 0) rectangle ++(0.4, 0.4)
                node[midway, non-path] {$\tmp$};
        }
    }

    \draw[->, non-path] (a11) -- (b1);
    \draw[->, non-path] (a11) -- (b2);
    \draw[->, active-path] (a12) -- (b3);
    \draw[->, non-path] (a12) -- (b4);
    
    \draw[->, non-path] (a21) -- (b5);
    \draw[->, non-path] (a21) -- (b6);
    \draw[->, non-path] (a22) -- (b7);
    \draw[->, non-path] (a22) -- (b8);

\end{tikzpicture}}
    \caption{
        \small \emph{(Left):} Three iterations of the Thorp shuffle~\cite{Thorp73} as computed by a network of switches. In a single iteration, we take cards $i$ and $n/2 + i$ and place them in positions $2i - 1$ and $2i$ in the order
        determined by a coin toss $r_j\in\{0,1\}$.
        \emph{(Right):} The shuffle can be simulated by a cell-probe algorithm (aka decision forest) that probes the coin tosses $r_j$~\cite[Lemma~6.4]{Viola12}. Drawn here is the decision tree that finds the $5$th output element.
    }
    \label{fig:thorp}
\end{figure}
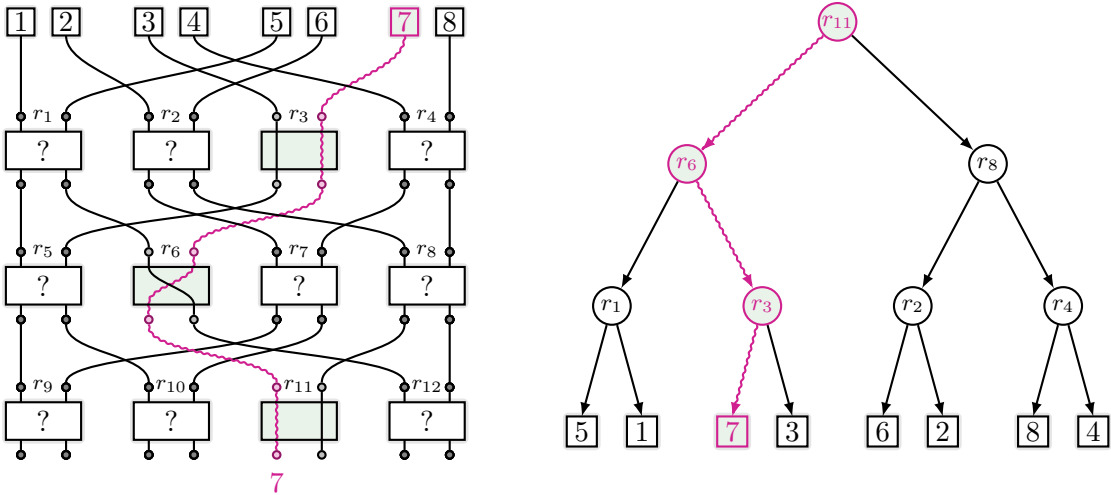

\pagebreak

\paragraph{Sampling permutations.}
We study cell-probe algorithms that solve \emph{sampling} problems. Here we are given a uniform random input $\bm u\sim [n]^s$ (even for infinite $s=\mathbb{N}\coloneqq\{0,1,2,\ldots\}$) and the goal is to produce a prescribed output distribution $f(\bm u)$. In this paper, we focus on producing an output distribution that is close to a uniform random permutation. More formally, we denote the set of permutations by $S_n\subseteq[n]^n$, a uniform random permutation by $\bm \pi\sim S_n$, and the statistical (total variation) distance between random variables $\bm x$ and $\bm y$ by $\Delta(\bm x, \bm y) \coloneqq \max_{E}|\!\Pr[\bm x \in E] - \Pr[\bm y \in E]|$.

Our research question becomes:

\begin{question} \label{q}
    What is the smallest depth $d$ of a decision forest $f$ such that $\Delta(\bm\pi,f(\bm u))\leq 1\%$?
\end{question}

Upper bounds on $d$ follow from oblivious shuffles. Indeed, \cref{fig:thorp} illustrates how $d$ iterations of the Thorp shuffle can be simulated by a depth-$d$ decision forest. It then follows from the aforementioned work on Thorp shuffle convergence that $d\leq O(\log^3 n)$ suffices. Better still, using an oblivious shuffle constructed by Czumaj~\cite{Czumaj15} one can obtain $d\leq O(\log^2 n)$.  (In fact, Czumaj conjectures that even $d\leq O(\log n)$ is possible.)

Lower bounds on $d$ are our main focus. The lower-bound question for sampling permutations with cell-probes was first raised by Viola~\cite[\S5]{Viola20}, who conjectured that $d\geq \omega(1)$ is necessary. In an accompanying seminar talk~\cite{Viola18}, Viola points out that, while $d=1$ is easy to rule out, existing techniques in the sampling literature (surveyed in~\cref{sec:survey}) do not even rule out $d=2$, surprisingly enough. Our main result is to confirm Viola's conjecture by proving the first non-trivial lower bounds on $d$. Moreover, our lower bound turns out to match the upper bounds from oblivious shuffles up to a polynomial factor.
\begin{theorem}[Main result]
\label{thm:main}
Suppose that $f\colon [n]^{\N} \to [n]^n$ is a decision forest of depth $(\log n)^{1/2-\epsilon}$ for some constant $\epsilon>0$. Then for $\bm u\sim[n]^{\N}$ and $\bm \pi \sim S_n$ we have
\[ \Delta(\bm \pi, f(\bm u)) \ge 1-\exp(n^{-\Omega(1)}). \]
\end{theorem}
The conclusion here gives a robust impossibility result, saying that the output of a shallow decision forest is extremely far from a uniform permutation. Such $1-o(1)$ distance bounds are known to imply lower bounds against succinct data structures for storing permutations. We discuss these corollaries in \cref{sec:data-structure}.

\paragraph{Nonadaptive algorithms.}
En route to \cref{thm:main} we develop new lower-bound techniques that are also able to show improved lower bounds against \emph{nonadaptive} algorithms. We say that a function $f\colon [n]^{\N} \to [n]^n$ is $k$-\emph{local} if every output cell depends on at most~$k$ input cells. That is, every $f_i$ makes at most~$k$ \emph{nonadaptive} probes to the input cells. Viola~\cite{Viola20} proved that permutation sampling requires $\Omega(\log \log n)$ nonadaptive probes, and this was subsequently improved to $\tOmega(\log n)$ by Yu and Zhan~\cite{YZ24}. Our second result gives another exponential improvement:
\begin{theorem}
\label{thm:main-nonadaptive}
Suppose that $f\colon [n]^s \to [n]^n$ is a $n^{\varepsilon}$-local function for a constant $\varepsilon \le 10^{-3}$. Then for~$\bm u \sim [n]^s$ and $\bm \pi \sim S_n$ we have
\[ \Delta(\bm \pi, f(\bm u)) \ge 1 - \exp(-n^{1-O(\varepsilon)}).\]
\end{theorem}
In particular, permutation sampling exhibits an $(\log n)^{O(1)}$-vs-$n^{\Omega(1)}$ separation between adaptive and nonadaptive cell-probe complexities. The only previous such separation for sampling problems was $O(1)$-vs-$\tOmega(\log n)$ by Yu and Zhan~\cite{YZ24}. The sampling problem they considered was somewhat artificial (defined for the sole purpose of obtaining a separation), whereas permutation sampling is an extremely natural problem. It remains open to prove an $O(1)$-vs-$n^{\Omega(1)}$ separation.

\subsection{Cell-probes vs.\ bit-probes}

For sampling permutations, it is important that we allow the full power of the cell-probe model, with input cells coming from a large alphabet $[n]$ corresponding to $O(\log n)$-bit word length. If we only allowed \emph{bit-probes}---that is, we restricted the input cells to be bits $\{0,1\}$---then sampling lower bounds would not meaningfully translate to data structure lower bounds (as discussed in~\cref{sec:data-structure}). In fact, bit-probe lower bounds for permutations are trivial to prove: Intuitively, each symbol in a random permutation is uniform over $[n]$ and so requires $\Omega(\log n)$ bits to generate, and this intuition is easy to turn into a formal proof of an $\Omega(\log n)$ bit-probe lower bound.\footnote{In a depth-$d$ \emph{binary} decision forest, each tree has range $2^d$. Thus the support size of the output of a forest is only~$2^{dn}$, which is negligible compared to $|S_n|=n!=2^{\Theta(n\log n)}$ when $d\leq o(\log n)$.} 

Many prior works have studied the bit-probe model. For example, by now, it is known that most distributions over $\{0,1\}^n$ that are \emph{symmetric} (invariant under permuting coordinates) cannot be sampled with $O(1)$ bit-probes~\cite{Viola23,FLRS23,KOW24,KOW25}. The key difference between bit-probes and cell-probes is that the output of a cell-probe algorithm need not be $k$-local for any nontrivial~$k$. Indeed, an $n$-ary decision tree of depth $d$ may depend on $n^{d-1}$ many input cells. This is why locality lower bounds even as high as $k\geq n^{\Omega(1)}$ do not rule out cell-probe algorithms with~$d=2$. Our results contribute new tools that get beyond this ``locality barrier''; such techniques are quite rare as only the prior work~\cite{Viola23} (discussed below in~\cref{sec:survey}) has specifically targeted the (adaptive) cell-probe model. Finally, we note that while there are even more powerful techniques to show sampling lower bounds against $\ACzero$-circuits (a model stronger than cell-probes), these do not apply for permutations as they are easy to sample for~$\ACzero$~\cite{Viola12}.

\subsection{Other related work} \label{sec:survey}
Our results contribute to the systematic study of the complexity of sampling distributions. Classically, computational complexity seeks hard-to-compute functions. In a seminal work, Viola~\cite{Viola12} proposed a program to find \emph{hard-to-sample distributions} for various computational models.
This program has been extremely fruitful for many areas, such as pseudorandom generators~\cite{Vio12b, LV12, BIL12}, randomness extractors~\cite{Viola14,CZ16,CS16}, error-correcting codes~\cite{SS24}, data-structure lower bounds~\cite{Viola12, LV12, BIL12, Viola20, CGZ22, Viola23, YZ24, KOW24, KOW25}, and learning theory~\cite{KOW26}. Especially for quantum--classical separations, a distribution is a natural witnessing object, and a line of work~\cite{BPP23,KOW24,GKMOW25} has shown quantum advantage for sampling problems.

The complexity landscape for sampling looks quite different from that for computation. Optimistically, one might hope to relate the complexity of computing $f\colon \{0,1\}^n \to \{0,1\}$ in some circuit class $\mathcal{C}$ to the complexity of sampling the input--output pair $(\bm u, f(\bm u))$ for $\bm u \sim \{0,1\}^n$ with a (multi-output) circuit from $\mathcal{C}$. Sampling input--output pairs is never harder than computing the function, but sometimes it can be dramatically easier.
The classical example is $\Xor(x) \coloneqq x_1 \oplus \dots \oplus x_n$. It is known since \cite{FSS84} that $\Xor$ is hard to compute with an $\ACzero$-circuit. On the other hand, $(\bm u, \Xor(\bm u))$ can be sampled with an $\textsf{NC}^0$-circuit \cite{Bab87}, in fact, with a $2$-local function
\[
y_1, \dots, y_{n} \mapsto y_1, y_1 \oplus y_2, y_2 \oplus y_3, \dots, y_{n-1} \oplus y_{n}, y_{n}.
\]
Kane, Ostuni, and Wu \cite{KOW25, KOW26} show that $(\bm u, \Xor(\bm u))$ is essentially the only non-trivial symmetric distribution that can be sampled with an $\textsf{NC}^0$-circuit. On the other hand, Viola \cite{Vio12b} building on the permutation-sampling algorithm by \cite{MV91,Hag91} shows that for every symmetric $f$ input--output pairs can be approximately sampled with an $\ACzero$-circuit.

The class $\ACzero$ is the frontier for sampling lower bounds: The works~\cite{LV12, BIL12} show that sampling a uniform codeword from a good error-correcting code is hard for an $\ACzero$-sampler. Viola~\cite{Viola20} exhibits a function $f$ whose input--output pairs cannot be sampled in $\ACzero$ with distance significantly smaller than $1/2$. 

Despite all this lower-bound progress for $\ACzero$-sampling, our understanding of the class remains quite coarse. Studying cell-probe algorithms, a model intermediate between~$\textsf{NC}^0$ and~$\ACzero$, can give us a more refined picture. The only prior work showing lower bounds specifically against (adaptive) cell-probe samplers is by Viola~\cite{Viola23}. He proves a \emph{separator theorem}, which states that any cell-probe sampler $f$ such that $f(\bm u)$ is close to a uniform distribution over some set, we can restrict the domain
of $f$ to some set $R$ such that $\bm y\coloneqq (f(\bm u) \mid \bm u \in R)$ retains most of its entropy, and the coordinates of $\bm y$ are \emph{almost pairwise independent} (the marginal distribution of every two output cells is close to the product of the individual margins). He used this theorem to separate $\ACzero$ and cell-probe algorithms for sampling problems. However, the separator theorem cannot immediately be used to prove lower bounds for sampling permutations: the coordinates of a permutation are almost pairwise independent, so there is no contradiction to $\bm y$ enjoying the same property.

\subsection{Succinct data structure lower bounds} \label{sec:data-structure}
Our results have direct implications for \emph{succinctly storing} permutations. A succinct data structure stores an object (for us, a permutation $\pi\in S_n$) with the number of bits close to the information theoretic minimum (for us, $\log n!$ bits), while supporting interesting queries. For permutations, it is natural to support querying the values $\pi(i)$ and possibly the inverses $\pi^{-1}(i)$. The paper~\cite{MRRR12} constructs a data structure for permutations with $\log n! + n/\log^{2-o(1)} n$ bits of memory that can answer both $\pi(i)$ and $\pi^{-1}(i)$-queries with $\tO(\log n)$ adaptive cell probes. For such a data structure, \cite{Golynsky09} gives an almost matching lower bound. 

Are there succinct data structures with better space/cell-probe complexity that only support~$\pi(i)$-queries? The best lower bound for this problem before this work followed from the sampling lower bound of \cite{YZ24}. They proved that every data structure for supporting $\pi(i)$-queries with $o(\log n / \log \log n)$ \emph{nonadaptive} probes must use at least $\log n! + n^{1-o(1)}$ bits. By a very simple reduction from \cite{Viola12} (also in \cite{Viola20}), our \cref{thm:main,thm:main-nonadaptive} imply the following.
\begin{corollary}
\label{cor:ds-lowerbound}
    Every data structure storing a permutation $\pi \in S_n$ and supporting $\pi(i)$-queries with 
    \begin{itemize}[noitemsep, topsep=0.2em]
        \item $(\log n)^{1/2-\epsilon}$ adaptive cell probes must use $\log n! + n^{\Omega(1)}$ bits of space;
        \item $n^{\varepsilon}$ nonadaptive cell probes must use $\log n! + n^{1-O(\varepsilon)}$ bits of space.
    \end{itemize}
\end{corollary}
We emphasize that such a lower bound is not obviously true at the outset, as surprising constructions of succinct data structures exist. For example, for the \emph{dictionary problem} of storing a set of key--value pairs, there exists a data structure that stores only a polylogarithmically many bits above the information-theoretic minimum with constant-time access~\cite{Yu2020}.

\subsection{Structure}
The rest of the paper is organized as follows: In \cref{sec:techniques} we review our techniques and give a full proof of \cref{thm:main-nonadaptive}. In \cref{sec:main-proof} we give the proof of \cref{thm:main} modulo three important technical tools. The three subsequent sections establish those tools: \cref{sec:lipschitzness} handles average Lipschitzness, the ingredient that differentiates \cref{thm:main} from \cref{thm:main-nonadaptive}; in \cref{sec:containment} we prove a \emph{containment lemma} that brings the statistical distance bounds exponentially close to $1$; in \cref{sec:proof-of-collision-lemma} we prove a \emph{collision lemma}, the main technical differentiator between \cref{thm:main-nonadaptive} and the previous work \cite{Viola20, YZ24}.  We conclude with open questions in \cref{sec:open-questions}.

\section{Techniques}
\label{sec:techniques}
In this section we discuss the proof of our main result. We start by giving in \cref{sec:nonadaptive-proof} a proof of \cref{thm:main-nonadaptive} with \emph{constant statistical distance bound}, and introducing the technical ideas shared between both of the proofs. In \cref{sec:main-proof-tech} we overview the additional ingredients needed for \cref{thm:main}. Finally in \cref{sec:boosting-distance} we show how to boost the statistical distance bound to be exponentially close to $1$ in both cases. 
\subsection{Proof of \cref{thm:main-nonadaptive}}
\label{sec:nonadaptive-proof}
In this section, we give a proof of \cref{thm:main-nonadaptive} modulo two technical lemmas. The informal
idea for both our proofs is the following dichotomy: if the (Shannon) entropy of $f(\bm u)$ is low, then a
random permutation~$\bm \pi\sim S_n$ lands in the support of $f(\bm u)$ with low probability, and if the entropy
of $f(\bm u)$ is high, then whp two of its symbols coincide, which never happens with a permutation.  Here Shannon entropy of a random variable $\bm x$ is $\H(\bm x) \coloneqq \sum_{x \in \supp(\bm x)} \Pr[\bm x = x] \log (1/\Pr[\bm x = x])$.

\paragraph{Low entropy case.} 
If we only shoot for a \emph{constant} statistical distance bound, we can get away with a very simple proof in this case without using any properties of the sampler apart from the entropy of $f(\bm u)$.
\begin{lemma}
\label{lem:simple-containment}
    Suppose that $\H(\bm x) \le k$. Then there exists a set $E$ of size $2^{2k}$ such that $\Pr[\bm x \in E] \ge 1/2$.
\end{lemma}
\begin{proof}
    Let $p(x) \coloneqq \Pr[\bm x = x]$ be the probability function of $\bm x$ so that $\H(\bm x) = \Exp[\log (1/p(\bm x))]$. We get from Markov's inequality that $\Pr[\log (1/p(\bm x)) \ge 2k] \le 1/2$. Thus for $E \coloneqq \{x \mid \log (1/p(x)) < 2k\}$ we have $\Pr[\bm x \in E] \ge 1/2$. Observe that $x \in E$ iff $p(x) > 2^{-2k}$, then since the total probability is at most $1$ we get $|E| \le 2^{2k}$.
\end{proof}
We can now conclude the low-entropy case almost immediately. If $\H(f(\bm u)) \le (n\log n)/4$, then by \cref{lem:simple-containment} we find $E$ of size $n^{n/2}$ such that $\Pr[f(\bm u) \in E] \ge 1/2$. On the other hand $\Pr[\bm \pi \in E] \le |E|/n! = o(1)$, which implies $\Delta(\bm\pi, f(\bm u)) \ge 1/2-o(1)$.

\paragraph{High entropy case.}
We would like to show that if $\H(f(\bm u)) > (n \log n)/4$, then some two symbols of $f(\bm u)$ coincide whp. The first step is to show this in the case the coordinates of $f(\bm u)$ are independent. We formalize this in the following \emph{collision lemma}:
\begin{restatable}{lemma}{simpcollisiona}
\label{lem:simplified-collision-lemma}
    Let $\bm z_1, \dots, \bm z_m$ be independent random variables over $[n]$ such that $\H(\bm z_1,\dots,\bm z_m) \ge (m \log n) / 8$ for $m \ge n^{0.99}$. Then $\Pr[\exists i \neq j \in [m]\colon \bm z_i = \bm z_j] \ge 1 - o(1)$.
\end{restatable}
This lemma (proved in \cref{sec:proof-of-collision-lemma}) and its generalizations will be one of the main technical ingredients in the proof of \cref{thm:main} as well.\footnote{It might seem that \cref{lem:simplified-collision-lemma} is a standard birthday paradox. Usually, the proofs of such results go as follows: split $\bm z$ into halves $\bm z_1, \dots, \bm z_{m/2}$ and $\bm z_{m/2+1}, \dots, \bm z_m$, fix $\bm z_{\le m/2} = \alpha$ and show that $\bm z_j \in A = \{\alpha_1,\dots,\alpha_{m/2}\}$ with noticeable probability for many $j > m/2$. Then apply Hoeffding's inequality. Observe that this is false, since it could happen that $\supp(\bm z_j) \cap A = \emptyset$ for all $j > m/2$.} 

\paragraph{Bounded influence.}
The remaining piece of the proof is the intermediate notion between local functions and a collection of independent output cells: we say that $f\colon [n]^s \to [n]^n$ is $(\ell, k)$-local if it is $k$-local and \emph{every input cell affects only $\ell$ output cells}. We make two simple observations. The first makes a step from $k$-locality to $(k^2,k)$-locality, and the second makes a step from $(\ell,k)$-locality to independence.
\begin{enumerate}[noitemsep, label=(O\arabic*)]
    \item For every $k$-local $f$ the distribution $f(\bm u)$ is a mixture of $n^{n/k}$ distributions $f^\alpha(\bm u)$ where $f^\alpha$ is $(k^2,k)$-local. Indeed, let $I \subseteq [s]$ be the set of inputs in $f$ that influence more than $k^2$ output cells. Since the number of input--output cell pairs such that the output cell depends on the input is at most $n \cdot k$, the size of $I$ is at most $n/k$. Hence for every $\alpha \in [n]^I$, fixing the inputs in $I$ according to $\alpha$ yields $f^\alpha$. 
    \label{obs:restriction}
    \item For every $(\ell, k)$-local function $f$ there is a set $J \subseteq [n]$ of size $n/(\ell k)$ of output cells such that $\{f_j(\bm u)\}_{j \in J}$ are independent. Indeed, populate $J$ greedily: add an output cell $j \in [n]$ to $J$ and ``delete'' all other output cells that share an input with $j$, and repeat this until all output cells are deleted. At each step at most $k\ell$ output cells are deleted and one of them is added to $J$, so $|J| \ge n/(\ell k)$.
    \label{obs:independence}
\end{enumerate}

\paragraph{Proof of \cref{thm:main-nonadaptive} with constant distance.} 
We will show that in the setting of \cref{thm:main-nonadaptive} we have $\Delta(f(\bm u), \bm \pi) \ge 1/2-o(1)$. We first apply \ref{obs:restriction} to get a set $I \subseteq [s]$ of size $n^{1-\epsilon}$ and a collection $\{f^\alpha\}_{\alpha \in [n]^I}$ of $(n^{2\epsilon}, n^\epsilon)$-local functions such that $(f(\bm u) \mid \bm u_I = \alpha) \equiv f^\alpha(\bm u)$. 
Consider an arbitrary $\alpha \in [n]^I$. The function $f^\alpha\colon [n]^{[s] \setminus I} \to [n]^n$ defined by restricting the inputs in $I$ according to $\alpha$ is still $n^\epsilon$-local and each input cell affects at most $n^{2\epsilon}$ output cells. 

We now apply the entropy dichotomy for $f^\alpha(\bm u)$:
\begin{enumerate}
\item  \emph{Low entropy case:} $\H(f^\alpha(\bm u)) \le (n \log n)/4$. Then by \cref{lem:simple-containment} there exists an event $E_\alpha$ such that $\Pr[f^\alpha(\bm u) \in E_\alpha] \ge 1/2$ with $|E_\alpha| \le n^{n/2}$. 

\item  \emph{High entropy case:} $\H(f^\alpha(\bm u)) \ge (n \log n)/4$. Then using the chain rule for Shannon entropy, we have $\sum_{i \in [n]} \H(f^\alpha_i(\bm u)) \ge \H(f^\alpha(\bm u)) \ge (n \log n) / 4$. Therefore, there exists a set $J \subseteq [n]$ of size $n/8$ such that for each $j \in J$ we have $\H(f^\alpha_j(\bm u)) \ge \log n / 8$. Now apply \ref{obs:independence} to find $J'$ of size at least $n^{1-3\epsilon}/8 \ge n^{0.99}$ such that $f^\alpha_j(\bm u)$ are independent for $j \in J'$. Applying \cref{lem:simplified-collision-lemma} to $f^\alpha_{J'}(\bm u)$ we get that $\Pr[\exists i \neq j \in J'\colon f_i(\bm u) = f_j(\bm u) \mid \bm u_I = \alpha] \ge 1 - o(1)$.
\end{enumerate}
Finally we are ready to define an event that witnesses the statistical distance between $f(\bm u)$ and $\bm \pi$. Suppose $L \subseteq [n]^I$ is the set of assignments $\alpha$ such that the entropy of $f^\alpha(\bm u)$ is low. Then define $F \coloneqq ([n]^n \setminus \supp(\bm \pi)) \cup \bigcup_{\alpha \in L} E_\alpha$. That is, $F$ is the event that the output sequence has a collision or belongs to one of the container events for the low-entropy case. Then $\Pr[\bm \pi \in F] \le n^{|I|} \cdot n^{n/2} / n! = o(1)$. On the other hand by the total probability law
\begin{align*}
  \Pr[f(\bm u) \in F] \ge
      & \Pr[\bm u_I \in L] \cdot 
          \Pr\left[f(\bm u) \in \bigcup_{\alpha \in L} E_\alpha \mid \bm u_I \in L\right]\\
                    &\, + \Pr[\bm u_I \not\in L] \cdot \Pr[f(\bm u) \text{ has a collision} \mid \bm u_I \not\in L]\\ \ge&\, 1/2.
\end{align*}

\paragraph{Comparison to \cite{Viola20}.}
The idea of low-vs-high entropy dichotomy was originally introduced in \cite{Viola20} to get $\Omega(\log \log n)$ locality lower bound for sampling permutations. The main reason our lower bound is much stronger is \cref{lem:simplified-collision-lemma}. In \cite{Viola20} the dichotomy was established for a much stronger notion of entropy. That caused the low-entropy case to be much more complicated, rendering \cref{lem:simple-containment} inapplicable.

\subsection{What is missing for the adaptive case?}
\label{sec:main-proof-tech}
The proof of \cref{thm:main} follows the same high-level recipe, but the steps are much more involved. The
proof of the nonadaptive case has three steps:

\begin{enumerate}[noitemsep, leftmargin = 5em, label=(Step \arabic*)]
    \item Fix some of the input cells so that the rest affect few output cells. \label{item:fixing}
    \item If the output has low entropy, argue that it is significantly contained in a small set. \label{item:low-entropy} 
    \item If the output has high entropy, argue that there is likely a collision by greedily
        choosing output cells that do not have common inputs. \label{item:high-entropy}
\end{enumerate}

The constant-probability version of \ref{item:low-entropy} does not suffer from the introduction of adaptivity, however the distance-boosting (see \cref{sec:boosting-distance}) in the adaptive case gets more involved. \ref{item:fixing} and \ref{item:high-entropy} break completely even in the constant-distance regime. The issue is that even a depth-$2$ decision forest can have arbitrary locality, and every input cell may influence every output cell. So, the key change is to devise an alternative intermediate notion between bounded-depth and independence. 

\paragraph{Average Lipschitzness.} 
Although low-depth decision trees are not local, they are local for every particular input. Similarly we can ask that a decision forest has bounded influence \emph{in expectation}:

\begin{restatable}[{\normalfont{\cite{BIL12}}}\mbox{}]{definition}{lipschitzdef} \label{def:lips}
    Let $f\colon [n]^s \to [n]^n$ be a decision forest. Let $\bm \theta_j$ denote the number of trees in $f$ that query $j$ on the input $\bm u \sim [n]^s$. We then say that $f$ is \emph{average-$\mu$-Lipschitz} if $\Exp[\bm \theta_j] \le \mu$ for every $j \in [s]$.
    We say that $f$ is $(\mu, \delta)$-Lipschitz if $\Pr[\bm \theta_j > \mu] \le \delta$ for every $j\in[s]$. 
\end{restatable}
The high-level plan of the proof of \cref{thm:main} is to implement \ref{item:fixing} and \ref{item:high-entropy} above with average Lipschitzness replacing bounded influence. Both adaptations come with challenges: In \ref{item:fixing} fixing input cells that have high average Lipschitzness might \emph{increase} the average Lipschitzness for other input cells. In \ref{item:high-entropy} the greedy choice of independent output cells as in \ref{obs:independence} fails completely\footnote{We remark that the greedy approach \emph{must fail}, since we do have depth-$O(\log^2 n)$ decision forests sampling permutations \cite{Czumaj15, Viola20}.}, since average Lipschitzness offers no \emph{global} independence properties. 

\paragraph{Average Lipschitzness implies Lipschitzness almost everywhere.}
At first, it seems that $(\mu, \delta)$-Lipschitzness (or Lipschitzness \emph{almost everywhere}) is a much stronger property than average Lipschitzness. Remarkably, it turns out that these properties are almost equivalent:
\begin{restatable}{lemma}{avgliptolip}
\label{lem:average-Lipschitz-to-Lipschitz}
    If $f\colon \Lambda^s \to \Sigma^m$ is an average-$\mu$-Lipschitz depth-$d$ decision forest,
    then for every $\varepsilon > 0$ the forest $f$ is $(3\mu d^2 \log (1/\varepsilon), \varepsilon)$-Lipschitz. 
\end{restatable}
This lemma constitutes the crucial structural property of Lipschitz forests that we exploit in our proofs. It implies that many functions of the output of a decision forest concentrate well around their expectations. In particular, we show this for the conditional entropy (\cref{subsec:cond-entr}), and the Hamming distance to a set (\cref{sec:containment}). 

The Lipschitz property of decision forests (with a different notation) was used very successfully by Beck, Impagliazzo, and Lovett \cite{BIL12}\footnote{They asked in Open~Problem~1 if
    their results could be used elsewhere. To the best of our knowledge, our work is the first such
    usecase.}. They prove \cite[Theorem~1.7]{BIL12} a result very similar (but incomparable) to \cref{lem:average-Lipschitz-to-Lipschitz}. Unfortunately, we cannot use their result directly, as it only implies $(\omega(\sqrt{s}), \cdot)$-Lipschitzness and we need the first parameter to be polynomially small in~$n$. 

\paragraph{Establishing average Lipschitzness.}
In the adaptive case we solve the issue with \ref{item:fixing} by fixing the inputs \emph{adaptively}. We exhaustively fix inputs that violate the condition $\Exp[\bm \theta_j] \le \mu$ and observe that \emph{in expectation over the value we fix the $j$-th input cell to}, this reduces the \emph{average depth} of the whole decision forest by $\mu$. Then we analyze the stopping time of the input fixing process to say that on average it terminates in $n d / \mu$ steps for a depth-$d$ decision forest. See \cref{sec:fixing-to-get-lipschitzness} for the proof.

\paragraph{Entropy-retaining depth reduction.}
For the nonadaptive case in \ref{item:high-entropy} we
used that after the restriction the input--output cell dependency graph has low degree. In the adaptive case,
even after establishing Lipschitzness almost everywhere, we still have potentially pairwise-dependent
output cells. We resolve this by further restricting the inputs \emph{and removing some output cells}
in order to reduce the depth of the forest. The key challenge here is to establish that after such
restriction the \emph{entropy rate} of the remaining output cells does not decrease much. Specifically we will
show that such procedure only decreases the entropy rate by a factor $1-O(1/d)$ when reducing the depth
of the forest from $d$ to $d-1$. Hence, after we are done, we are going to have some $m$ output cells that are
computed with a depth-$1$ decision forest and with entropy $\Omega(m \log n / d)$. Depth-$1$ decision
forest is nonadaptive, so we can again choose independent output cells greedily.  Our procedure shrinks the
number of output cells by a $\poly(\log n)$-factor at each step, so in order to apply (generalized)
\cref{lem:simplified-collision-lemma} we need $\poly(\log n)^d = o(n^{1/d})$, hence we choose $d =
o(\sqrt{\log n}/\log\log n)$.

In order to show that whp over the input restrictions the entropy is retained we show that conditional
Shannon entropy $\H(f(\bm u) \mid \bm u_I = \bm\alpha)$ concentrates around $\H(f(\bm u) \mid \bm u_I)$
wrt the random choice of $\bm\alpha \sim [n]^I$ when $f$ is an average-Lipschitz decision forest. The
corresponding property is almost immediate for \emph{min-entropy} regardless of the structure of the
random variable $f(\bm u)$. However, for us it is crucial to work with Shannon entropy, since assuming
that min-entropy is low does not imply any \emph{global} properties of the distributions, so the analogue
of \cref{lem:simple-containment} for min-entropy is completely false.

\subsection{Boosting distance}
\label{sec:boosting-distance}
\cref{thm:main-nonadaptive} and \cref{thm:main} both claim that the statistical distance from a cell-probe sampler to $\bm \pi$ is exponentially close to $1$. This bound is crucial for the application to succinct data structures. In this section, we show how to boost the distance in the nonadaptive case.

In the high-entropy case we actually do not need any changes, since \cref{lem:simplified-collision-lemma} implies that probability of \emph{not} seeing a collision in a high-entropy collection of independent variables is exponentially low. 
Thus, it remains to address the low-entropy case where we want to show that for a $(n^{2\epsilon},n^\epsilon)$-local function $g$ there exists an event $E$ such that $|E| \le n^{3n/4}$ and $\Pr[g(\bm u) \in E] \ge 1 - \exp(-n^{1-10\epsilon})$. 

Let $F$ be the event given by \cref{lem:simple-containment}: $\Pr[g(\bm u) \in F] \ge 1/2$ and $|F| \le n^{n/2}$. The main idea is to consider a \emph{neighborhood} of $F$ as the new witnessing event: $\mathcal{N}_k(F) \coloneqq \{x \mid \min_{y\in F}\dist(x, y) \le k\}$, where $\dist(\cdot, \cdot)$ is the ($n$-ary) Hamming distance.

The probability bound $\Pr[g(\bm u) \in F] \ge 1/2$ is equivalent to  $|g^{-1}(F)| \ge n^s/2$. If $x \in \mathcal{N}_k(g^{-1}(F))$ then $g(x) \in \mathcal{N}_{k \cdot n^{2\epsilon}}(F)$, since changing one symbol in the input to $g$ changes at most $n^{2\epsilon}$ output cells. We then need to choose $k$ such that $\mathcal{N}_k(g^{-1}(F))$ contains almost all points of $[n]^s$, yet $E \coloneqq \mathcal{N}_{k \cdot n^{2\epsilon}}(F)$ is small as shown in \cref{fig:low-entropy-error-boosting}.
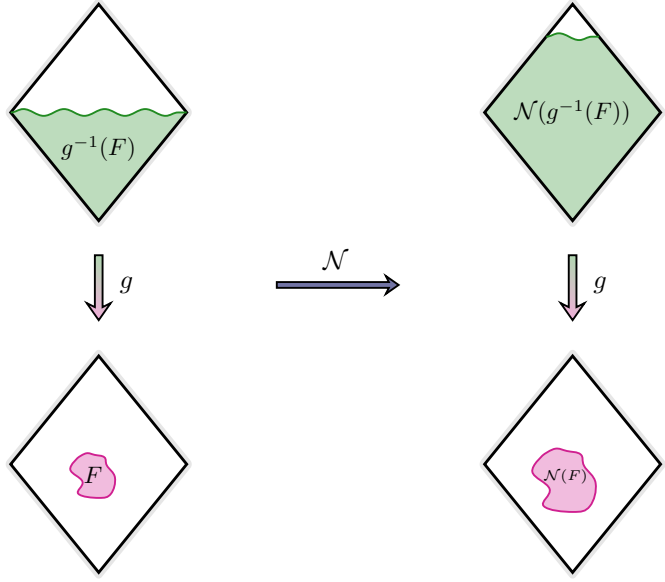
\begin{figure}[h]
    \centering
    \scalebox{0.9}{\begin{tikzpicture}
    \def\a{1.3}
    \def\b{1.6}
    \def\rat{\b / \a}
    \def\sh{7}
    
    \foreach \point in {(0, 0), {(0, -2 * \b - 2)}, (\sh, 0), {(\sh, -2 * \b - 2)}}{
        \begin{scope}[shift = (\point)]
            \fill[black!10, rounded corners = 2pt] (-0.1 - \a, 0) -- (0, {-\b - 0.1 * \rat}) --
                (\a + 0.1, 0) -- (0, {\b + 0.1 * \rat}) -- cycle;
            \fill[white] (-\a, 0) -- (0, -\b) -- (\a, 0) -- (0, \b) -- cycle;
        \end{scope}
    }
    \foreach \point in {{(0, -2 * \b - 2)}, {(7, -2 * \b - 2)}}{
        \begin{scope}[shift = (\point)]
            \draw[very thick] (-\a, 0) -- (0, -\b) -- (\a, 0) -- (0, \b) -- cycle;
        \end{scope}
    }
    \foreach \x in {0, \sh}{
        \begin{scope}[shift = {(\x, 0)}]
            \node[fancy-arrow, scale = 0.4] (ar) at (0, -2 * \b + 0.65) {};
            \node[right = 5pt] at (ar) {$g$}; 
        \end{scope}
    }

    \node[fancy-arrow, scale = 0.3,
        minimum height = 6cm, rotate = 90,
        right color = Periwinkle!50!black,
        left color = Periwinkle!50!black]
        (ar) at (\sh / 2, -2 * \b + 0.65) {};
    \node[above = 3pt] at (ar) {$\mathcal{N}$};

    \fill[ForestGreen!30!white] (0, -\b)  -- (-\a, 0)
        decorate[decoration = {snake, segment length = 6mm, amplitude = 0.5mm}]{ -- (\a, 0)} -- cycle;
    \draw[very thick] (-\a, 0) -- (0, -\b) -- (\a, 0) -- (0, \b) -- cycle;
    \draw[ForestGreen, thick, decorate, decoration = {snake, segment length = 6mm, amplitude = 0.5mm}]
        (-\a, 0) -- (\a, 0);
    \node at (0, -\b / 3) {\small $g^{-1}(F)$};
    \begin{scope}[shift = {(\sh, 0)}]
        \path[very thick, name path = fr] (-\a, 0) -- (0, -\b) -- (\a, 0) -- (0, \b) -- cycle;
        \path[name path = temp] (-\a, {(0.7 * \b)}) -- ++(2 * \a, 0);
        \path[name intersections = {of = fr and temp}];
        \fill[ForestGreen!30!white] (0, -\b)  -- (-\a, 0) -- (intersection-2)
            decorate[decoration = {snake, segment length = 6mm, amplitude = 0.5mm}]{ -- (intersection-1)}
            -- (\a, 0) -- cycle;
        \draw[very thick] (-\a, 0) -- (0, -\b) -- (\a, 0) -- (0, \b) -- cycle;
        \draw[ForestGreen, thick, decorate, decoration = {snake, segment length = 6mm, amplitude = 0.5mm}]
            (intersection-2) -- (intersection-1);

        \node at (0, 0) {\small $\mathcal{N}(g^{-1}(F))$};
    \end{scope}

    \begin{scope}[shift = {(0, -2 * \b - 2)}]
        \draw[thick, VioletRed, fill = VioletRed!30!white, scale = 0.12, shift = {(1.5, -1.65)}]
            plot[smooth, tension = 0.7] coordinates {
                (-4, 2.5) (-3, 3) (-2, 2.8) (-0.8, 2.5) (-0.5, 1.5) (0.5, 0)
                (0, -2) (-1.5, -2.5) (-4, -2) (-3.5, -0.5) (-5, 1) (-4, 2.5)
            };
        \node at (-0.08, -0.16) {\small $F$};
    \end{scope}

    \begin{scope}[shift = {(\sh, -2 * \b - 2)}]
        \draw[thick, VioletRed, fill = VioletRed!30!white, scale = 0.17, shift = {(1.5, -1.65)}]
            plot[smooth, tension = 0.7] coordinates {
                (-4, 2.5) (-3, 3) (-2, 2.8) (-0.8, 2.5) (-0.5, 1.5) (0.5, 0)
                (0, -2) (-1.5, -2.5) (-4, -2) (-3.5, -0.5) (-5, 1) (-4, 2.5)
            };
        \node at (-0.1, -0.18) {\tiny $\mathcal{N}(F)$};
    \end{scope}

\end{tikzpicture}}
    \caption{\small The picture illustrates the process of boosting the distance from $1/2$ to almost $1$ by expanding the witnessing event to its neighborhood.}
    \label{fig:low-entropy-error-boosting}
\end{figure}

For the former, we use a fact from combinatorics, that is essentially due to Harper \cite{Harper1966}.
\begin{theorem}
\label{thm:harper-mcd}
    For an arbitrary set $S \subseteq [n]^s$ and $\bm u \sim [n]^s$  we have 
    \[\Pr[\bm u \in \mathcal{N}_k(S)] \ge 1 - \frac{\exp(-k^2/(2s\log n))}{\Pr[\bm u \in S]}.\]
\end{theorem}
\begin{proof}
    For the boolean alphabet ($n=2$) the claim is shown in \cite[Proposition~7.7]{McDiarmid89}. We give a simple reduction to this case. Consider the natural bijection $b\colon [n]^s\to \{0,1\}^{s \log n}$ that encodes every $[n]$-symbol as $\log n$ bits. Then $\Pr[\bm u \in \mathcal{N}_k(S)] \ge \Pr[b(\bm u) \in \mathcal{N}_{k}(b(S))]$, which implies the claim.
\end{proof}
In order to achieve high probability in \cref{thm:harper-mcd} we choose $k = n^{1-4\epsilon}$. Wlog we may assume that $s\le n^{1+\epsilon}$ for a $n^\epsilon$-local function. Hence we get $\Pr[\bm u \in \mathcal{N}_k(g^{-1}(F))] \ge 1 - \exp(-n^{1-9 \epsilon}/2\log n) / 2$. On the other hand $|\mathcal{N}_{k \cdot n^{2\epsilon}}(F)| \le |F| \cdot \binom{s}{k \cdot n^{2\epsilon}} n^{k n^{2\epsilon}}  \le n^{n/2} \cdot \exp(3 k \cdot n^{2\epsilon} \log s) \le n^{n/2} \cdot n^{3 n^{1-2\epsilon}} \ll n^{3n/4}$. 

\paragraph{Boosting distance in the adaptive case: containment lemma.}
 The problem with the proof of \ref{item:low-entropy} above is that in the adaptive case the input space might have the dimension up to $n^{d+1}$, which prevents us from using \cref{thm:harper-mcd} to boost the error probability: the radius of the neighborhood that we would have to use is at least $\sqrt{n^{d+1}}$, which is \emph{larger than the dimension of the output space}. Thus we prove a dimension-free version of \cref{thm:harper-mcd} for sets that can be recognized by bounded-depth decision trees: 
\begin{theorem}[simplified version of \cref{lem:coupling}]
\label{thm:coupling-simplified}
    Suppose that $T\colon [n]^s \to \{0, 1\}$ is a decision tree of depth $d$ such that $|T^{-1}(1)|
    \ge \mu \cdot n^s$. Then
    \[
        \Exp_{\bm u \sim [n]^s}[\dist(\bm u, T^{-1}(1))] = O(\sqrt{d \log (1/\mu)}).
    \]
\end{theorem}
For every depth-$d$ decision forest \emph{every property} of its output can be computed with a depth-$nd$ decision tree, so with \cref{thm:coupling-simplified} we can reduce the effective dimension to $nd$. On the other hand, the expectation bound we get is not quite enough to establish the exponentially low probability of being outside the neighborhood. Thus, in \cref{sec:containment} we boost \cref{thm:coupling-simplified} using average Lipschitzness.

\section{Proof of the \nameref{thm:main}}
\label{sec:main-proof}
In this section, we give the proof of the lower bound for the adaptive case and formally introduce all the technical tools needed. 

\subsection{Warm-up: the case of bucketed queries}
We start by proving the theorem in the special case where the queries of the decision forest are structured: every tree takes its first query from a set $I_1$, the second query from the set $I_2$, and so on $[s] = I_1 \sqcup \dots \sqcup I_d$. We say that such a forest is \emph{bucketed}. The goal of this section is to present the high-level structure of the general proof and introduce the primary technical tools that will be utilized throughout. We remark that instead of the domain $[n]^{\mathbb{N}}$ used in the intro, we use a domain $[n]^s$ for some integer $s$: a depth-$d$ decision forest (bucketed or not) consisting of $n$ trees can query at most $n^{d+1}$ distinct input cells, so we may assume wlog that the indices of these cells are $[s]$. 
\begin{theorem}
\label{thm:warmup-main}
Suppose $f\colon [n]^s \to [n]^n$ is a bucketed depth-$o(\log n/\log \log n)$ decision forest with buckets $I_1, \dots, I_d$ and let $\bm u_i \sim [n]^{I_i}$ for each $i \in [d]$. Let $\bm \pi \sim S_n$ be the uniform random permutation. Then
\[ \Delta(f(\bm u_1, \dots, \bm u_d), \bm \pi) \ge 1 - \exp(-n^{\Omega(1)}).\]
\end{theorem}

As announced in \cref{sec:main-proof-tech}, on the high level, this proof follows the plan of the proof of \cref{thm:main-nonadaptive} with bounded influence replaced with bounded average Lipschitzness (see \cref{def:lips}). Let us first consider two cases: when $\H(f(\bm u)) \ge n \log n / 4$ and when $\H(f(\bm u)) \le n \log n / 4$.

\paragraph{High entropy case.}
Suppose that $\H(f(\bm u)) \ge n \log n / 4$. Ideally, in this case we want to show that probability of a collision is 
\[
    \Pr[\exists i \neq j \in [n]\colon f_i(\bm u) = f_j(\bm u)] \ge 1 - \exp(-n^{\Omega(1)}),
\]
and thus, almost always, we will not generate a permutation. The main idea is to reduce the problem to the case where the output cells of $f$ are independent, and apply a collision lemma similar to \cref{lem:simplified-collision-lemma}. Here we will use the following version:
\begin{restatable}{lemma}{simpcollisionb}
\label{lem:simplified-collision-lemma2}
    Let $\bm u \sim [n]^s$ and $f\colon [n]^s \to [n]^n$ be a depth-$1$ decision forest. Suppose that \(
    \H(f(\bm u)) \ge 4(n \log \log n).
    \)
    Then
    \( \Pr[\exists i \neq j \in [n] \colon f_i(\bm u) = f_j(\bm u)] \ge 1 - \exp(-\Omega(n/\poly(\log n))).\)
\end{restatable}
Observe that with a stronger guarantee on the entropy this would follow from \cref{thm:main-nonadaptive}, since depth-$1$ decision forests are $1$-local functions. We prove this lemma in \cref{subsec:depth-1}.

In the bucketed case the reduction from a depth-$d$ to a depth-$1$ decision forest is quite straightforward: \emph{fix all input buckets, except for one.} Thus, we need to find a bucket $I_i$ such that after fixing inputs in all other buckets the output of the resulting forest \emph{typically} retains some entropy. The following simple fact says that after such conditioning the entropy is retained \emph{in expectation}:
\begin{restatable}{fact}{entaftcond}
\label{fact:entropy-with-condition}
Suppose $\bm x_1, \dots, \bm x_\ell$ are independent and $\bm y = f(\bm x_1, \dots, \bm x_\ell)$ then 
\(\H(\bm y) \le \sum_{i \in [\ell]} \H(\bm y \mid \bm x_{[\ell] \setminus i}).\)
\end{restatable}
The proof of this fact is a simple chain rule computation and can be found in \cref{sec:entropy-fact-proof}.
We then find that there exists $i \in [d]$ such that
\begin{equation}\label{eq:entropy-lower-bound}
    \H(f(\bm u) \mid \bm u_{[d] \setminus i}) \ge (n \log n / 4)/d \gg n \log \log n.
\end{equation}
As observed before, for any $\beta \in [n]^s$ and $j \in [d]$, function $x \mapsto f(\beta_{[s] \setminus I_j}, x_{I_j})$ can be computed with a depth-$1$ decision forest, hence if we proved that
\begin{equation}
    \Pr_{\bm\beta \sim [n]^{s}}\left[\H(f(\bm u) \mid \bm u_{[d] \setminus i} = \bm\beta_{[s] \setminus I_i}) \ge 4(n \log \log n)\right] \ge 1 - \exp(-n^{\Omega(1)}),
    \label{eq:target-entropy-bound}
\end{equation}
we would be able to immediately apply \cref{lem:simplified-collision-lemma2} and finish the proof in the high-entropy case. 

Unfortunately, \eqref{eq:entropy-lower-bound} does not imply the probability lower bound by itself. However, for an \emph{average Lipschitz} forest $f$, retaining a high entropy under restriction \emph{typically} is implied by having a high conditional entropy:
\begin{restatable}{lemma}{simplifiedentconc}
    \label{lem:simplified-entropy-concentration}
        Let $f\colon [n]^s \to \Sigma^m$ be an average-$n^{0.3}$-Lipschitz depth-$\log^{O(1)}n$ decision forest. Suppose that $m \le n$ and $|\Sigma| = O(n)$. Let $I \subseteq [s]$ be a subset of the input cells. Then for $\bm r$ and $\bm u$ uniformly distributed over $[n]^s$ we have
    \[ \Pr\left[|\H(f(\bm u) \mid \bm u_I = \bm r_I) - \H(f(\bm u) \mid \bm u_I)| > n^{0.9}\right] \le \exp(-n^{\Omega(1)}).\]
\end{restatable}
We prove this lemma in \cref{subsec:cond-entr}. Combined with \eqref{eq:entropy-lower-bound} this lemma implies \eqref{eq:target-entropy-bound} finishing the proof in the high-entropy case for average Lipschitz functions. We will see how to get rid of this assumption shortly.
\paragraph{Low entropy case.} 
Now suppose that $\H(f(\bm u)) \le n \log n / 4$. Observe that here, analogously to the nonadaptive case in \cref{thm:main-nonadaptive} the would have been resolved with \cref{lem:simple-containment} had we aimed for the constant distance bound. For the strong bounds we need to boost the distance. To this end, we prove the following lemma in \cref{sec:containment}:
\begin{restatable}{lemma}{containmentlemma}
\label{lem:containment}
    Suppose $f\colon[n]^s\to [n]^m$ is an average-$n^{0.1}$-Lipschitz depth-$d$ decision forest with $m \geq n^{0.99}$, $d = \log^{O(1)} n$, $\bm u \sim [n]^s$, and  $\H(f(\bm u))\leq m\log n / 4$. Then there exists a set $F \subseteq [n]^m$ of size $n^{3m/4}$ such that
    \[ \Pr[f(\bm u) \in F] \ge 1 - \exp(-n^{\Omega(1)}).\]
\end{restatable}

\paragraph{Final step: enforcing average Lipschitzness.} The only obstacle that remains for the proof of \cref{thm:warmup-main} is the failure of average Lipschitzness for $f$. Similar to the nonadaptive case we enforce it by restricting the values of input cells that are queried too many times in expectation.  

We need some additional notation to state this result. Each leaf of a decision tree can be naturally identified with a partial assignment $\ell \in ([n] \cup \{\star\})^s$, where the non-$\star$ symbols correspond to the input cells queried in the path to the leaf and $\star$-symbols correspond to all the remaining symbols. A random leaf $\bm \ell$ of a decision tree is defined as the leaf reached by the computation on the random input $\bm u \sim [n]^s$. For a decision forest $f\colon [n]^s \to [n]^m$ and a partial assignment $\ell$, $f|_\ell$ denotes the decision forest $f$ after restricting the input cells according to $\ell$.
    \begin{restatable}{lemma}{strongavglipschitz}
    \label{lem:strong-avg-Lipschitz}
        Suppose $f\colon [n]^s \to [n]^m$ is an arbitrary depth-$d$ decision forest. There exists a $2nd/\mu \cdot \log(1/\varepsilon)$-depth decision tree $T$ querying symbols of $[n]^s$ such that for a random leaf $\bm \ell$ of $T$, $f|_{\bm \ell}$ is average-$\mu$-Lipschitz with probability $1-\varepsilon$.
    \end{restatable}We then can derive a simple corollary (along the lines of the proof of \cref{thm:main-nonadaptive}) that immediately implies \cref{thm:warmup-main}.
    \begin{corollary}
    \label{cor:dichotomy}
        Suppose that $f\colon [n]^s \to [n]^n$ is such that for any partial assignment $\ell \in ([n] \cup \{\star\})^s$ that makes $f|_\ell$ average-$n^{0.1}$-Lipschitz, and for $\bm u \sim [n]^s$ we have either
        \begin{itemize}[noitemsep, topsep=2pt]
            \item \emph{(Collision):} $\Pr[\exists i \neq j \in [n]\colon (f|_\ell)_i(\bm u) = (f|_\ell)_j(\bm u)] \ge 1 - \exp(-n^{\Omega(1)})$.
            \item \emph{(Containment):} There is a set $F_\ell$ of size $n^{3n/4}$ such that $\Pr[f|_\ell(\bm u) \in F_\ell] \ge 1 - \exp(-n^{\Omega(1)})$.
        \end{itemize}
        Then $\Delta(f(\bm u), \bm \pi) \ge 1 - \exp(-n^{\Omega(1)})$, where $\bm \pi \sim S_n$.
    \end{corollary}
    \begin{proof} We apply \cref{lem:strong-avg-Lipschitz} with $\epsilon=\exp(-n^{.01})$ and $\mu=n^{.1}$ to get a depth-$n^{.99}$ decision tree. Consider a leaf $\ell$ of this tree. By the assumption for each $\ell$ such that $f|_\ell$ is average-$\mu$-Lipschitz, we either get a collision with probability $1 - \exp(-n^{\Omega(1)})$, or there exists a set $F_\ell$ of size at most $n^{3n/4}$ such that $\Pr[f(\bm u) \in F_\ell \mid \bm u \in \ell] \ge 1 - \exp(-n^{\Omega(1)})$. The total number of leaves is bounded by $n^{n^{.99}} \ll n^{n/20}$. Thus, the size of the union of all $F_\ell$ over all leaves that fall in the containment case is at most $n^{n/20} \cdot n^{3n/4} \le n^{4n/5}$.

    Finally, we see that with probability $1 - \exp(-n^{\Omega(1)})$ a uniformly random $\bm u$ lands in a leaf $\bm \ell$ such that $f|_{\bm \ell}$ is average-$n^{.1}$-Lipschitz. If that happens, then either the sequence that we sample does not constitute a permutation, or it belongs to $\bigcup F_\ell$. Given that $\Pr[\bm \pi \in \cup_{\ell} F_\ell] \le n^{4n/5}/n! = \exp(-\Omega(n))$, this immediately implies that
    \( \Delta(f(\bm u), \bm \pi) \ge 1 - \exp(-n^{\Omega(1)}).\)
    \end{proof}

\subsection{The general proof}
In this section we finalize the proof of \cref{thm:main}. The only missing structural piece is the treatment of the high-entropy case. It is formalized as follows:
\begin{lemma}
\label{lem:high-entropy-case}
    Suppose $f\colon [n]^s \to [n]^m$ is a depth-$d$ average-$\mu$-Lipschitz decision forest such that $d = o(\sqrt{\log n}/\log \log n)$, $\mu = n^{0.1}$, and $m = \Omega(n)$. Suppose that $\H(f(\bm u)) \ge m \log n / 4$.
    Then $\Pr[\exists i \neq j \in [m]\colon f_i(\bm u) = f_j(\bm u)] \ge 1 - \exp(-n^{\Omega(1)})$. 
\end{lemma}

Assuming this lemma, we can finish the proof.
\begin{proof}[Proof of \cref{thm:main}]
 Consider a partial assignment $\ell \in ([n] \cup \{\star\})^s$, suppose that $f|_\ell$ is average-$n^{0.1}$-Lipschitz. Then either $\H(f|_\ell(\bm u)) \ge n \log n / 4$, in which case we have a collision whp by \cref{lem:high-entropy-case}, or $\H(f|_\ell(\bm u)) \le n \log n / 4$, so we have a small container set by \cref{lem:containment}. We now apply \cref{cor:dichotomy} and finish the proof.
\end{proof}

We now proceed to the proof of \cref{lem:high-entropy-case}.

\subsubsection{Sharper tools}
So far, we have stated two simplified versions of our collision lemma: \cref{lem:simplified-collision-lemma} and \cref{lem:simplified-collision-lemma2}. The general adaptive case requires a stronger version (proved in \cref{sec:proof-of-collision-lemma}):
\begin{restatable}{lemma}{depthonecollision}
\label{lem:collision-for-d-1}
    Let $\bm u \sim [n]^s$ and $f\colon [n]^s \to ([n] \cup \{\bot\})^m$ be a depth-$1$ decision forest. Suppose that $\H(f(\bm u)) \ge \delta \cdot m \log n$, and $(\delta^2/4) m \ge n^{1-\epsilon}$ for $\delta = \delta(n)  \ge \max(4\log\log n /\log n, 8\epsilon)$. Then
    \[ \Pr[\exists i \neq j \in [m] \colon f_i(\bm u) = f_j(\bm u) \neq \bot] \ge 1 - \exp(-\Omega(\delta^4 m^3/n^2)).\]
\end{restatable}

In the warm-up section we hid under the rug the need to use $(\mu, \delta)$-Lipschitzness, instead of average Lipschitzness. Although by \cref{lem:average-Lipschitz-to-Lipschitz} we can always assume average Lipschitzness, operating with almost-everywhere Lipschitzness directly comes in handy in the general proof, because the latter is preserved under restrictions:
\begin{lemma}
\label{lem:Lipschitzness-after-conditioning}
    Suppose $f\colon \Lambda^s \to \Sigma^m$ is a $(\mu, \delta)$-Lipschitz decision forest.
     Let $T$ be a decision tree querying symbols of a string in $\Lambda^s$.
     For $\bm \alpha$ be a leaf of $T$ where a uniformly random assignment lends,
      and let $f|_{\bm\alpha}$ 
     be the forest with the input cells restricted according to $\bm\alpha$. Then
    \[ \Pr[\text{ $f|_{\bm\alpha}$ is $(\mu, \sqrt{\delta})$-Lipschitz }]
        \ge 1 -\sqrt{\delta}.\]
\end{lemma}
\begin{proof}
    Fix some $j \in [s]$. Let $Q(\alpha, x)$ be the number of trees in 
    $f|_{\alpha}$ that query the cell $j$ on the input $x$. Let $E$ be the event 
    ``$f|_{\bm\alpha}$ is \emph{not} $(\mu, \sqrt{\delta})$-Lipschitz''. 
    Suppose for contradiction that $\Pr[E] > \sqrt{\delta}$.
    Consequently we get
    \[ \Pr_{\bm \alpha, \bm u}[Q(\bm \alpha, \bm u) > \mu]
     \ge \Pr[E] \cdot \Pr[Q(\bm \alpha, \bm u) > \mu \mid E] 
     > \sqrt{\delta} \cdot \sqrt{\delta} = \delta.\]
    Since $Q(\alpha, x)$ is also the number of trees in $f$ that query the cell $j$ on the 
    joint input $(\alpha, x)$ we have a contradiction with $(\mu, \delta)$-Lipschitzness 
    of $f$.
\end{proof}

A helpful trick for dealing with Lipschitz decision forests is to turn an unlikely undesirable event (say an input cell is queried too many times) into an impossible event by terminating the computation of a tree if it is about to do something undesirable (e.g. query that too popoular input cell). The following fact helps to argue that such terminations do not affect entropy too much. 
\begin{fact}
\label{fact:mixture-of-rv}
Let $\bm x$ be a random variable supported over $(\Sigma \cup \{\bot\})^n$ and let $b\colon (\Sigma \cup \{\bot\})^n \to \mathbb{Z}_{\ge 0}$ be the function that counts non-$\bot$ elements in the input. Then $\H(\bm x) \le \log (n+1) + \Exp[b(\bm x)] \log (n |\Sigma|)$.
\end{fact}
\begin{proof}
    Observe that $\H(\bm x \mid b(\bm x) = \ell) \le \ell \log(n |\Sigma|)$, since the support size of $(\bm x \mid b(\bm x) = \ell)$ is at most $(n|\Sigma|)^\ell$. We then write by the chain rule:
    \begin{align*}
        \H(\bm x) &= \H(b(\bm x), \bm x) \\
                  &= \H(b(\bm x)) + \H(\bm x \mid b(\bm x)) \\
                  &= \H(b(\bm x)) + \Exp_{\bm \ell \sim b(\bm x)}[\H(\bm x \mid b(\bm x) = \bm \ell)]\\
                  &\le \H(b(\bm x)) + \Exp_{\bm \ell \sim b(\bm x)}[\bm \ell \log(n |\Sigma|)] \\
                  &\le \log (n+1) + \Exp[b(\bm x)] \cdot \log (n |\Sigma|).\qedhere
    \end{align*}
\end{proof}

\subsubsection{Handling the high-entropy case.}
In this section, we prove \cref{lem:high-entropy-case}.
The proof proceeds by repeatedly reducing the depth of the forest until the depth is $1$, so we are in a position to apply \cref{lem:collision-for-d-1}. 

In what sense do we reduce the depth? We will find a set $I \subseteq [s]$ and a set $J \subseteq [m]$ such that whp over the assignment to $I$ the projection $f_J$ after assigning the input cells in $I$ has high entropy and depth at most $d-1$. This is formalized in the following key lemma:
\begin{lemma}
\label{lem:one-step-lemma}
    Suppose $f\colon [n]^s \to ([n] \cup \{\bot\})^m$ is a depth-$d$ average-$\mu$-Lipschitz decision forest with $m \ge n^{0.99}$, and $\mu \le n^{0.3}$. Suppose that for $\bm u \sim [n]^s$ we have $\H(f(\bm u)) \ge c m \log n$ with $c = \omega(1/\log n)$. Then there exists a set $I \subseteq [s]$, a set $J \subseteq [m]$ of size $|J| \ge m / \log^6(n)$ and a forest $g\colon [n]^s \to ([n] \cup \{\bot\})^J$ such that for every $j \in J$ and all $x \in [n]^s$ we have $g_j(x) \in \{f_j(x), \bot\}$ and one of the following conditions holds:
    \begin{enumerate}[noitemsep, label=(C\arabic*)]
            \item With probability $1-\exp(-n^{\Omega(1)})$ over $\bm\beta \sim [n]^I$, $\H(g(\bm u) \mid \bm u_I = \bm \beta) \ge (1-3/d) c |J| \log n$ and $g$ has depth $d-1$ after any assignment to the input cells in $I$. \label{item:decrease-by-1}
            \item With probability $1-\exp(-n^{\Omega(1)})$  over $\bm\beta \sim [n]^I$, $\H(g(\bm u) \mid \bm u_I = \bm \beta) \ge (1  / (3d)) c |J| \log n$ and $g$ has depth $1$ after any assignment to the input cells in $I$. \label{item:drop-to-1}
    \end{enumerate}
\end{lemma}

\begin{proof}[Proof of \cref{lem:high-entropy-case} given \cref{lem:one-step-lemma}]

We apply \cref{lem:average-Lipschitz-to-Lipschitz} to get that for every $\delta > 0$ we have that $f$ is $(3\mu d^2 \log(1/\delta), \delta)$-Lipschitz. We take $\delta \coloneqq \exp(-\mu)$, so $f$ is $(3(\mu d)^2, \delta)$-Lipschitz.
Then we will iterate \cref{lem:one-step-lemma} until the restricted $f$ is depth-$1$ according to the following algorithm.
\begin{algorithmic}[1]
\Input{$f \colon [n]^s \to ([n] \cup \{\bot\})^m$}
\Output{$f' \colon [n]^{R} \to ([n] \cup \{\bot\})^{J}$ of depth $1$. A partial assignment $\bm \alpha$ to $[n]^s$}
\vspace{0.5em}
\State Let $f' \gets f$ with $R \gets [s]$ and $J \gets [m]$.
\State $\bm \alpha$ be an empty assignment.
\While{Depth of $f'$ is larger than $1$}
\State Apply \cref{lem:one-step-lemma} to get $g\colon [n]^R \to ([n] \cup \{\bot\})^{J'}$, $J' \subseteq J$ and $I \subseteq R$.
\State Sample $\bm\beta \sim [n]^I$.
\State $\bm \alpha \gets \bm \alpha \cup \bm \beta$.
\State Update $f' \gets g|_{\bm \beta}$; $R \gets R \setminus I$; $J \gets J'$.
\If{$f'$ is not $((3\mu d)^2, \sqrt{\delta})$-Lipschitz}
   \State {\bf fail} \label{line:fail}
\EndIf
\EndWhile
\end{algorithmic}
By \cref{lem:Lipschitzness-after-conditioning} the line \eqref{line:fail} is executed at any point in the algorithm with probability at most $\sqrt{\delta}$ over $\bm \alpha$. With probability $(1-\exp(-n^{\Omega(1)}))^d$ over the random assignment $\bm \alpha$ the entropy of $f'(\bm u)$ satisfies for some $j \in [d]$
\[ \frac{\H(f'(\bm u))}{|J|\log n} \ge \prod_{i=j}^d (1-O(1/i)) \cdot \Omega(1/j) \ge \Omega(1/d) = \omega(\log\log n/\sqrt{\log n}).\]
On the other hand $|J| \ge m(\log^6 n)^{-d} = n \cdot 2^{o(\sqrt{\log n})}$. 

Now we reduced the problem to the case $d=1$, so we can apply \cref{lem:collision-for-d-1} to get that $f_J$ has a non-$\bot$ collision with probability $1-\exp(-n^{\Omega(1)})$ over $\bm u \sim [n]^R$. Hence with probability $(1-\sqrt{\delta})(1-\exp(-n^{\Omega(1)}))$ there is a collision in $f(\bm \alpha, \bm u)$ as required.
\end{proof}

\subsubsection{Proof of \cref{lem:one-step-lemma}} The proof goes as follows:
\begin{enumerate}[noitemsep]
    \item We are going to identify sets $I \subseteq [s]$ and $J \subseteq [m]$ such that the \emph{first
        queries} in trees $f_J$ always come from $I$ and only some $o(|J|)$ trees query $I$ after their
        first query whp over the input. \label{item:first-step}
    \item We will prune the trees $f_J$ into $g_J$ such that $g_J$ \emph{never} query $I$ after the first
        query at the expense of sometimes returning $\bot$. \label{item:second-step}
    \item By \cref{fact:entropy-with-condition} we have that either $\H(g_J(\bm u) \mid \bm u_I) \ge
        \H(g_J(\bm u)) \cdot (1-1/d)$ or $\H(g_J(\bm u) \mid \bm u_{[s] \setminus I}) \ge \H(g_J(\bm u))
        \cdot (1/d)$, so in the former case \ref{item:decrease-by-1} is satisfied and in the latter
        \ref{item:drop-to-1} is. \label{item:third-step}
\end{enumerate}
In the steps \eqref{item:first-step} and \eqref{item:second-step} we need to make sure that the entropy
\emph{rate} of $g_J(\bm u)$ does not drop too much compared to the entropy \emph{rate} of $f(\bm u)$. In
the step \eqref{item:third-step} the key is to use the fact that the conditional entropy
\emph{concentrates} i.e. not only conditioning on a \emph{random variable $\bm u_I$} does not reduce it
too much, but conditioning on the event $\bm u_I = \bm\beta$ also does not reduce it too much whp over
$\bm\beta \sim \bm u_I$.

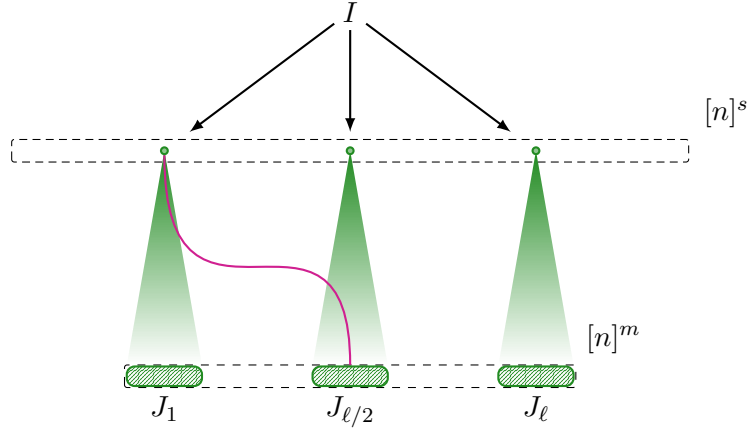
\begin{figure}
    \centering
    \tikzset{
    rectangle-set/.style = {
        #1,
        pattern = {Lines[angle = 45, distance = 1pt, line width = 0.2pt]},
        pattern color = #1,
        thick,
        rounded corners = 3pt
    }
}
\begin{tikzpicture}
    \def\a{4.5}
    \def\b{3}
    \def\d{3}

    \draw[thin, densely dashed, rounded corners = 1pt] (-\a, \d) rectangle (\a, \d + 0.3);
    \draw[thin, dashed, rounded corners = 1pt] (-\b, 0) rectangle (\b, 0.3);

    \node at (\a + 0.5, \d + 0.65) {$[n]^s$};
    \node at (\b + 0.5, 0.65) {$[n]^m$};

    \node[inner sep = 3pt] (I) at (0, \d + 2) {$I$};
    \path[name path = temp] (-\a, \d + 0.4) -- ++(2 * \a, 0);

    \foreach \x [count = \i] in {0.03, \b - 0.5, 4.97}{
        \draw[rectangle-set = {ForestGreen}] (-\b + \x, 0.02) rectangle ++(1, 0.26);
        \node[in-out-node = {ForestGreen}] (a\i) at (-\b + \x + 0.5, \d + 0.15) {};

        \fill[ForestGreen, path fading = south] (a\i) -- (-\b + \x, 0.3) -- ++(1, 0) -- (a\i);

        \path[name path = a] (I) -- (a\i);
        \path [name intersections = {of = temp and a, by = tt}];
        \draw[thick, ->] (I) -- (tt);
    }

    \draw[thick, VioletRed] (a1) .. controls ++(0, -3) and (0, 3) .. (0, 0.3);

    \node[text centered, text depth = 0.1cm, minimum height = 0.5cm] at (-\b + 0.03 + 0.5, -0.3)
        {$J_{1}$};
    \node[text centered, text depth = 0.1cm, minimum height = 0.5cm] at (0, -0.3) {$J_{\ell / 2}$};
    \node[text centered, text depth = 0.1cm, minimum height = 0.5cm] at (-\b + 4.97 + 0.5, -0.3)
        {$J_{\ell}$};
    
    \node at (-\a - 2, 0) {};
    \node at (\a + 2, 0) {};
\end{tikzpicture}
    \caption{\small The picture illustrates the choice of $\bm J$. We first sample the set $\bm I \subseteq [\ell]$ of input cells and take to $\bm J$ only the trees that query a symbol in $\bm I$ as their first query. The undesirable events for us are \textbf{non-first} queries from $\bm J$ to $\bm I$, it is represented as the red line in the picture. The main point is that for a fixed $i \in [\ell]$ and $j \in [m]$ this happens with probability $\alpha^2$, but the expected size of $\bm J$ is an $\alpha$-fraction of $[m]$. Hence, the subsampling procedure sparsifies the undesirable events.}
    \label{fig:distinguishing-first-queries}
\end{figure}

\paragraph{First step: isolating first queries.} 
First, we are going to choose some trees in $f_i$ such that the set of input cells $I$ that are queried first by $f_i$ is unlikely to be queried as a non-first query by \emph{any} of the chosen trees. 

Let us partition $f_1, \dots, f_m$ into subsets $J_1 \sqcup \dots \sqcup J_\ell = [m]$ such that trees indexed with $J_i$ first query the input cell $i$ (wlog we may assume that the first queries form the set $[\ell]$). By the Lipschitzness assumption we have that $|J_i| \le 2\mu$ for every $i \in [\ell]$, hence $\ell \ge m/(2\mu)$. Now let $\bm I \subseteq [\ell]$ be a random set where each element from $[\ell]$ is included independently with probability $\alpha = 1/\log^6 n$. 
Let $\bm J \coloneqq \bigcup_{i \in \bm I} J_i$ be the set of output cells that first query an input cell from $\bm I$. 

We would like to argue that the expectation over $\bm I$ of the expectation over $\bm u$ of the number of non-first queries $f_{\bm J}$ make to $\bm I$ is low, see \cref{fig:distinguishing-first-queries} for the illustration. Let $p_{ij} \coloneqq \Pr[f_i(\bm u) \text{ queries } j]$. Then
\[ \Exp\left[\sum_{a \neq b \in \bm I} \sum_{(i,j) \in J_a \times J_b} p_{ij}\right] = \sum_{a \neq b \in [\ell]} \Pr[a \in \bm I] \Pr[b \in \bm I] \sum_{(i,j) \in J_a \times J_b} p_{ij} \le \alpha^2 md. \]
We now argue that there exists $I \subseteq [\ell]$ and corresponding $J \coloneqq \bigcup_{i \in I} J_i$ satisfying three conditions:
\begin{enumerate}[noitemsep, label=(R\arabic*)]
    \item $\H(f_J(\bm u)) \ge (1 - 1/d) \cdot \alpha \H(f(\bm u))$. \label{item:retaining-entropy}
    \item $|J| \le (1 + 1/d) \cdot \alpha m$. \label{item:not-too-many-retained}
    \item $\Exp[\text{ number of non-first queries to $I$ that $f_J$ make }] \le \alpha^{-1/2} d \cdot \alpha^2 md$. \label{item:expectation}
\end{enumerate}
\begin{claim}
    Conditions \ref{item:retaining-entropy}-\ref{item:expectation} are satisfied by $\bm I$ and $\bm J$ with positive probability.
\end{claim}
\begin{proof} 
 By Markov's inequality $\bm I$ satisfies the condition \ref{item:expectation} with probability $1-\sqrt{\alpha}/d$. 
 
  Let us compute the probability that \ref{item:not-too-many-retained} is not satisfied by $\bm J$. For an event $A$, let $\lb A \rb$ denote the random variable that is $1$ if $A$ occurs and $0$ otherwise. Then
\begin{align*} \Pr[|\bm J| > (1+1/d) \alpha m] &= \Pr\left[\sum_{i \in [\ell]} |J_i| \lb i \in \bm I \rb > (1+1/d) \alpha m\right]\\
&= \Pr\left[\sum_{i \in [\ell]} |J_i| \lb i \in \bm I \rb - \Exp[|\bm J|] > 1/d \cdot \alpha m\right]\\
\text{(Hoeffding's inequality) }&\le \exp(-\Omega((\alpha/d)^2 \cdot m^2 / \sum_{i \in [\ell]} |J_i|^2))\\
&= \exp(-\Omega(\alpha^2 m/(\mu d)^2))\\
\text{(since $m \ge n^{0.99}$, and $\mu \le n^{0.3}$) }&=\exp(-n^{\Omega(1)}).
\end{align*}
Now we turn to \ref{item:retaining-entropy}. The key step is to show that the entropy is retained in expectation: $\Exp_{\bm J}[\H(f_{\bm J}(\bm u))] \ge \alpha \H(f(\bm u))$. This holds by Shearer's inequality:
\begin{lemma}[Shearer's Inequality \cite{CGFS86}]
    \label{lem:shearer}
    Suppose $\bm S \subseteq [n]$ is a distribution over subsets of $[n]$ such that for every $i \in [n]$ $\Pr[i \in \bm S] \ge \kappa$. Then for any random variable $\bm x \sim \Sigma^n$ we have
    \[ \H(\bm x) \le \frac{1}{\kappa} \Exp[H(\bm x_{\bm S})].\]
\end{lemma}

On the other hand
\begin{align*}
\Exp\nolimits_{\bm J}[\H(f_{\bm J}(\bm u))] 
\le&\;\hphantom{+} \Pr[\H(f_{\bm J}(\bm u)) < (1-1/d) \cdot \alpha \H(f(\bm u))] \cdot (1 - 1/d) \cdot \alpha \H(f(\bm u)) 
\\&+ \Pr[\H(f_{\bm J}(\bm u)) \ge (1-1/d) \cdot \alpha \H(f(\bm u)) \land |\bm J| \le 2\alpha m] \cdot 2 \alpha m \log n \\
&+ \Pr[|\bm J| > 2 \alpha m] \cdot \H(f(\bm u))\\
\le&\, (1-1/d) \cdot \alpha \H(f(\bm u)) + p \cdot 2 \alpha m \log n + \exp(-n^{\Omega(1)}),
\end{align*}
where $p \coloneqq \Pr[\H(f_{\bm J}(\bm u)) \ge (1 - 1/d) \cdot  \alpha \H(f(\bm u))]$. By rearranging we get
\[ \frac{\H(f(\bm u)) / d - \exp(-n^{\Omega(1)})}{2 m \log n} \le p. \] 
 Since $\H(f(\bm u)) = c m \log n$ we then get $p \ge c/3d = \omega(1/(d \log n))$, so since $\omega(1/(d \log n))- \sqrt{\alpha}/d - \exp(-n^{\Omega(1)}) = \omega(1/(d \log n)) > 0$ there exists $I$ satisfying \ref{item:retaining-entropy}, \ref{item:not-too-many-retained}, \ref{item:expectation}. 
\end{proof}

\paragraph{Second step: pruning the trees.}
Having $I \subseteq [\ell]$ and $J \subseteq [m]$ satisfying the conditions \ref{item:retaining-entropy}-\ref{item:expectation} we now define trees $g_j\colon [n]^s \to [n] \cup \{\bot\}$ for each $j \in J$ as follows: $g_j$ follows the behavior of $f_j$ until it is about to query an input cell from $I$ in which case it returns $\bot$. 
\begin{claim} 
\label{claim:small-drop-when-adding-bot}
$\H(g_J(\bm u)) \ge \H(f_J(\bm u)) - O(\alpha^{3/2} m d^2 \log n) \ge (1-1/d) \H(f_J(\bm u))$. 
\end{claim}
\begin{proof}
    The second inequality is satisfied since 
    \[\alpha^{3/2} m d^2 \log n = \alpha m d^2 / \log^2 n \le \alpha m / \log n = o(\alpha c m \log n / d) = o(\H(f_J(\bm u)) / d).\]
    
    We now prove the first inequality. Consider a random variable $\bm b \in ([n] \cup \{\bot\})^J$ such that $\bm b_j = \bot$ if $g_j(\bm u) \neq \bot$ and $\bm b_j = f_j(\bm u)$ if $g_j(\bm u) = \bot$. Then $f_J(\bm u)$ is uniquely determined by $\bm b$ and $g_J(\bm u)$, so $\H(f_J(\bm u)) \le \H(g_J(\bm u)) + \H(\bm b)$. Then, since the expected number of non-$\bot$ symbols in $\bm b$ is bounded by \ref{item:expectation}, we conclude since by \cref{fact:mixture-of-rv} 
    \(\H(\bm b) \le  \log n (2 + 2\alpha^{3/2} m d^2).\)
\end{proof}
\paragraph{Third step: restricting the inputs.}
By \cref{fact:entropy-with-condition} we either have $\H(g_J(\bm u) \mid \bm u_I) \ge \H(g_J(\bm u)) \cdot (1-1/d)$ or $\H(g_J(\bm u) \mid \bm u_{[s] \setminus I})\ge \H(g_J(\bm u)) \cdot (1/d)$. Suppose that the former holds, then, since $g$ is average-$\mu$-Lipschitz and $\mu \le n^{0.3}$, \cref{lem:simplified-entropy-concentration} implies
\begin{align*}
 \Pr_{\bm r \sim \bm u}[\H(g_J(\bm u) \mid \bm u_I = \bm r_I) \ge (1-2/d) \H(g_J(\bm u))] &\ge\\
  \Pr_{\bm r \sim \bm u}[|\H(g_J(\bm u) \mid \bm u_I = \bm r_I) -\H(g_J(\bm u) \mid \bm u_I)| \le (1/d) \H(g_J(\bm u))]  &\ge\\
  \Pr_{\bm r \sim \bm u}[|\H(g_J(\bm u) \mid \bm u_I = \bm r_I) -\H(g_J(\bm u) \mid \bm u_I)| \le n^{0.9}]&\ge 1 - \exp(-n^{\Omega(1)}).
 \end{align*}
 The second inequality holds since $\H(g_J(\bm u)) / d \ge m / \poly(\log n) \gg n^{0.9}$. If $\H(g_J(\bm u) \mid \bm u_{[s] \setminus I})\ge \H(g_J(\bm u)) \cdot (1/d)$ holds we apply \cref{lem:simplified-entropy-concentration} analogously. Then either \ref{item:decrease-by-1} or \ref{item:drop-to-1} is satisfied since $\H(g_J(\bm u)) \ge (1-1/d) \H(f_J(\bm u))$ by \cref{claim:small-drop-when-adding-bot}.

\subsection{Entropy after assignment}
\label{sec:entropy-fact-proof}
In this section, we prove \cref{fact:entropy-with-condition}, the proof is a simple chain rule computation. 

\entaftcond*
We first prove it in the special case of $\ell=2$:
\begin{fact}
\label{fact:entropy-with-condition2}
Suppose $\bm x_1$, $\bm x_2$, and $\bm z$ are independent and $\bm y = f(\bm x_1, \bm x_2, \bm z)$ then 
\[\H(\bm y \mid \bm z) \le \H(\bm y \mid \bm x_1, \bm z) + \H(\bm y \mid \bm x_2, \bm z).\]
\end{fact}
\begin{proof} We write
    \begin{align*}
        \H(\bm x_1 \mid \bm z) + \H(\bm x_2 \mid \bm z) &= \H(\bm x_1, \bm x_2 \mid \bm z) \\&= \H(\bm x_1, \bm x_2, \bm y \mid \bm z)\\
      \text{(chain rule for Shannon entropy) }  &= \H(\bm y \mid \bm z) + \H(\bm x_1 \mid \bm y, \bm z) + \H(\bm x_2 \mid \bm y, \bm x_1, \bm z) \\
      \text{(entropy decreases with conditioning) }  &\le \H(\bm y \mid \bm z) + \H(\bm x_1 \mid \bm y, \bm z) + \H(\bm x_2 \mid \bm y, \bm z)\\
        &=\H(\bm y, \bm x_1 \mid \bm z) + \H(\bm y, \bm x_2 \mid \bm z) - \H(\bm y \mid \bm z) 
    \end{align*}
By  applying $\H(\bm y, \bm x_i \mid \bm z) = \H(\bm y \mid \bm x_i, \bm z) + \H(\bm x_i \mid \bm z)$ for $i \in \{1,2\}$ and rearranging we get the claim.
\end{proof}
Now we can derive \cref{fact:entropy-with-condition}.
\begin{proof}[Proof of \cref{fact:entropy-with-condition}]
Applying \cref{fact:entropy-with-condition2} with empty $\bm z$ we get:
\[
    \H(f(\bm x))  \le \H(f(\bm x) \mid \bm x_1) + \H(f(\bm x) \mid \bm x_{[d] \setminus 1}).
\]
Then we continue to rewrite $\H(f(\bm x) \mid \bm x_1) \le \H(f(\bm x) \mid \bm x_1, \bm x_2) + \H(f(\bm x) \mid \bm x_1, \bm x_{[n] \setminus \{1,2\}})$ again using \cref{fact:entropy-with-condition2} with $\bm z = \bm x_1$. We then get the claim by a simple induction.
\end{proof}

\section{Lipschitz Decision Forests}
\label{sec:lipschitzness}
Two of our technical lemmas are in fact some form of concentration, \cref{lem:simplified-entropy-concentration} is precisely the concentration of conditional entropy, and \cref{lem:containment} can be seen as the concentration of \emph{distance from a random point to a set}. In both cases we reduce the question to the McDiarmid's inequality.

\paragraph{McDiarmid's inequality.}
A function $f\colon \Lambda^s \to \mathbb{R}$ has \emph{$c$-bounded differences over $S \subseteq \Lambda^s$} if for every $x,x' \in S$ that differ in one coordinate we have $|f(x) - f(x')| \le c$.

The following concentration inequality is well-known and can be found, for example, in \cite{Combes15}. 

\begin{lemma}[McDiarmid's inequality]
\label{lem:mc-diarmid}
    Suppose $f\colon \Lambda^s \to \mathbb{R}$ has $c$-bounded differences over $S$ and for $\bm u \sim [n]^s$ we have $\Pr[\bm u \in S] = 1-\delta$. Then
    \[ \Pr[|f(\bm u) - \Exp[f(\bm u) \mid \bm u \in S]| \ge \lambda + \delta c s] \le 2(\delta + \exp(-\Omega(\lambda^2/(c^2 s))).\]
\end{lemma}

The bounded difference property is very similar in spirit to Lipschitzness, the topic of this section. We restate the definition for convenience:
\lipschitzdef*

This concept was studied (with somewhat different notation) by Beck, Impagliazzo, and Lovett \cite{BIL12}, who proved a concentration inequality for Lipschitz forests.
\begin{theorem}[{\cite[Theorem~1.7]{BIL12}}]
\label{thm:avg-Lipschitz-concetration}
Suppose that $f\colon \{0,1\}^s \to \{0,1\}^m$ is an average-$\mu$-Lipschitz depth-$d$ decision forest. Then
\[ \Pr_{\bm u \sim \{0,1\}^s}\left[\Big|\sum_{i \in [n]} f_i(\bm u) - \Exp\big[\sum_{i \in [n]} f_i(\bm u)\big]\Big| > d \cdot \sqrt{\mu \cdot n \cdot \log(d^4/\varepsilon)}\right] \le \varepsilon.\]
\end{theorem}
Unfortunately, we cannot use their results in a black-box way, the first reason is that they are stated only for the binary alphabet, whereas we need it for \emph{exponential-size} alphabet, and the second is because of the $\sqrt{n}$ multiplier in the deviation bound. We would like to apply such concentration bound in the case of low expectation, so this multiplier would be too weak. For these reasons we prove the following simpler deviation bound:
\begin{lemma}
\label{lem:second-moment}
    Suppose that $f\colon \Lambda^s \to [0,1]^m$ is a depth-$d$ average-$\mu$-Lipschitz decision forest.
    Then if $\kappa \coloneqq \Exp\left[\sum_{i \in [m]} f_i(\bm u)\right]$, we have for every $\varepsilon > 0$
    \[\Pr\left[\sum_{i \in [m]} f_i(\bm u) \ge 2(\kappa + \log (1/\varepsilon) d \mu)\right] \le \varepsilon.\]
\end{lemma}

As a corollary of \cref{lem:second-moment} we get, following the simplified proof of \cite[Corollary~1.8]{BIL12} that

\avgliptolip*
\begin{proof}
    As above, for every $j \in [s]$ let $\bm \theta_j$ be number of trees in $f$ querying $j$ on the input $\bm u$.
    Let $g^{\ell,j}_i$ be the depth-$\ell$ decision tree that returns $1$ if the $\ell$-th query of $f_i$ is to $j$ and $0$ otherwise. Clearly $g^{\ell, j}$ is obtained by pruning $f$ up to the $\ell$-th layer and replacing all $j$-labels with $1$ and all others with $0$. 
    
    Then $\bm \theta_j = \sum_{\ell \in [d]} \sum_{i \in [n]} g^{\ell,j}_i(\bm u)$. The forest $g^{\ell, j}_i$ for fixed $j$ and $\ell \in [d]$ , $i \in [n]$ is average-$d\mu$-Lipschitz, since every tree in the forest is obtained by pruning a tree in $f$ and every tree from $f$ corresponds to $d$ trees in the forest. Thus, by \cref{lem:second-moment} we have 
    \[ \Pr[\bm \theta_j \ge 2(\mu + \log(1/\varepsilon) d^2 \mu)] \le \varepsilon,\]
    which implies the claim.
\end{proof}

\subsection{Proof of \cref{lem:second-moment}}
     We denote $\tau_i(\bm u) \in (\Lambda \times [s])^d$ the transcript of running $f_i$ on $\bm u$ (the queries and the outcomes). We abuse notation to denote by $f_i(\alpha)$ the value of $f_i$ given the transcript $\alpha \in (\Lambda \times [s])^d$. 

     We now estimate the $k$-th moment of $\sum_{i \in [m]} f_i(\bm u)$ by induction on $k$. We are going to prove the statement by induction on $k$. Let $F(\kappa, k)$ be the maximum of $\Exp[(\sum_{i \in [m]} h_i(\bm u))^k]$ over all depth-$d$ average-$\mu$-Lipschitz decision forests $h$ with $\Exp[\sum_{i \in [m]} h_i(\bm u)] = \kappa$. The base of induction is $k=1$ where trivially $F(\kappa, 1) = \kappa$.
    \begin{align*}
        \Exp\left[\Big(\sum_{i \in [m]} f_i(\bm u)\Big)^k \right]
        &= 
        \sum_{i,j_1,\dots,j_{k-1} \in [m]} \Exp[f_i(\bm u) f_{j_1}(\bm u) \cdot \dots \cdot f_{j_{k-1}}(\bm u)]\\
        &=
        \sum_{i \in [m]} \sum_{\alpha \in (\Lambda \times [s])^d} \Pr[\tau_i(\bm u) = \alpha] f_i(\alpha) \cdot \Exp\left[\Big(\sum_{j \in [m]} f_j(\bm u)\Big)^{k-1} \,\middle|\, \tau_i(\bm u) = \alpha\right]\\
     \end{align*}
     Let us now estimate $\Exp\left[(\sum_{j \in [m]} f_j(\bm u))^{k-1} \mid \tau_i(\bm u) = \alpha\right]$ for fixed $i,\alpha$.
    Consider a decision forest $g^\alpha$ such that $g_i^\alpha$ is a copy of $f_i$ where all queries to $\alpha$ are replaced with leaves labeled with $1$, in other words, when we are about to query something in $\alpha$, we return $1$ instead. Such a transformation will preserve the $\mu$-Lipschitz property, as for every $j$, the number of queries can only decrease. We then further estimate the expectation
    \begin{align*}
        \Exp\left[\Big(\sum_{j\in [m]} f_j(\bm u)\Big)^{k-1} \,\middle|\, \tau_i(\bm u) = \alpha\right] &\le \Exp\left[\Big(\sum_{j\in [m]} g_j^\alpha(\bm u)\Big)^{k-1} \,\middle|\, \tau_i(\bm u) = \alpha\right] \\
      \text{ (as $g_j^\alpha$ never query $\alpha$) } &=
        \Exp\left[\Big(\sum_{j\in [m]} g_j^\alpha(\bm u)\Big)^{k-1}\right] \\
        &\le F\left(\Exp\left[\sum_{j \in [m]} g^\alpha_j(\bm u)\right], k-1\right)
    \end{align*}
    Now let us compute the new expectation:
    \begin{align*}
         \Exp\left[\sum_{j\in [m]} g_j^\alpha(\bm u)\right] &= 
         \Exp\left[\sum_{j\in [m]} (g_j^\alpha(\bm u) - f_j(\bm u))\right] + \kappa \\
         &\le d \mu + \kappa.
    \end{align*}
    Here we use that $(g_j^\alpha(\bm u) - f_j(\bm u)) \neq 0$ only if $f_j$ queries $\alpha$, hence the expectation of the sum of these over $j \in [n]$ is bounded by the expected number of queries to $\alpha$, which is at most $d\mu$. Putting all together we get
    \begin{align*} \Exp\left[\Big(\sum_{i \in [m]} f_i(\bm u)\Big)^k\right] &\le F(d \mu + \kappa, k-1) \cdot \sum_{i \in [m];\; \alpha \in (\Lambda \times [s])^d} f_i(\alpha)\Pr[\tau_i(\bm u) = \alpha]\\
    &\le F(d \mu + \kappa, k - 1) \cdot \kappa\\
    &\le (\kappa + k d \mu)^k.
    \end{align*}
    Finally, by Markov inequality
     \[\Pr\left[\sum_{i \in [m]} f_i(\bm u) \ge a\right] = \Pr\left[\Big(\sum_{i \in [m]} f_i(\bm u)\Big)^k \ge a^k\right] \le \frac{\Exp\left[\Big(\sum_{i \in [m]} f_i(\bm u)\Big)^k\right]}{a^k} \le \left(\frac{\kappa + kd\mu}{a}\right)^k\]
    and we get the desired inequality by substituting $a = 2(\kappa + d\mu\cdot\log (1/\varepsilon))$ and $k = \log (1/\varepsilon)$.

\subsection{Conditional entropy concentration}\label{subsec:cond-entr}
The main property of Lipschitz forests is that conditional entropy concentrates in the following sense:
\begin{lemma}
\label{lem:first-entropy-concentration}
    Suppose $f\colon [n]^s \to \Sigma^m$ is an $(\mu, \delta)$-Lipschitz depth-$d$ decision forest with $\mu = m^{\Theta(1)}$ and $\delta = \exp(-m^{\Omega(1)})$. Let $I \subseteq [s]$ be a subset of the input cells. Then for $\bm r$ and $\bm u$ uniformly distributed over $[n]^s$ we have
    \[ \Pr[|\H(f(\bm u) \mid \bm u_I = \bm r_I) - \H(f(\bm u) \mid \bm u_I)| > \lambda \log(m|\Sigma|) \sqrt{\mu m d}] \le \exp(-\Omega(\lambda^2)) + \exp(-m^{\Omega(1)}).\]
\end{lemma}
Before proceeding to the proof of this lemma, let us derive a simpler form that we use in \cref{sec:main-proof}:
\simplifiedentconc*
\begin{proof}
    Wlog we may assume $m=n$ by adding $n-m$ trivial decision trees to $f$. Then, taking $\lambda=n^{0.1}$, we get $\lambda \log(m |\Sigma|) \sqrt{\mu m d} = \poly(\log n) \cdot n^{0.75} \ll n^{0.9}$. 
\end{proof}

The key step in the proof of \cref{lem:first-entropy-concentration} is to show that the conditional entropy has bounded differences:
\begin{lemma}
\label{lem:entropy-deviation}
    Suppose $f\colon \Lambda^s \to \Sigma^m$ is an average-$\mu$-Lipschitz decision forest.
    Then for $\bm y \sim \Lambda$, $\bm z \sim \Lambda^{s-1}$, and every $y \in \Lambda$ we have with $\bm o \coloneqq f(\bm y, \bm z)$ that
    \[ |\H(\bm o \mid \bm y = y) - \H(\bm o)| \le \log (m+1)+ \mu \log (m|\Sigma|). \]
\end{lemma}
\begin{proof}
    Fix some $y \in \Lambda$ and let $\bm o' \coloneqq f(y, \bm z)$. Let $\bm a \in (\Sigma \cup \{\bot\})^m$ be the random variable such that for each $i \in [m]$ we have $\bm a_i = \bm o'_i$ if the first input cell was queried by $f_i$ and $\bm a_i = \bot$ if it was not. Then $\H(\bm o') \le \H(\bm o, \bm a) \le \H(\bm o) + \H(\bm a)$, since $\bm o'$ is uniquely determined by $\bm o$ and $\bm a$. We then observe that the expected number of non-$\bot$ elements in $\bm a$ is at most $\mu$ in expectation: indeed $f$ is average-$\mu$-Lipschitz and fixing the first input cell does not change the expected number of queries to it, since the trees in $f$ query each symbol at most once. Then by \cref{fact:mixture-of-rv} we get $\H(\bm a) \le\log (m+1) + \mu \log(m |\Sigma|) $
    Hence $\H(\bm o \mid \bm y = y) \le \H(\bm o) + \log (m+1) + \mu \log(m |\Sigma|)$. The reverse direction is proved analogously.
\end{proof}

\begin{corollary}
\label{cor:entropy-bd-differences}
    Suppose $f \colon \Lambda^s \to \Sigma^m$ is a $(\mu, \delta)$-Lipschitz decision forest, with $\delta = \exp(-m^{\Omega(1)})$, and $\mu = m^{\Theta(1)}$. Let $I \subseteq [s]$ be a subset of input cells, and let a function $h\colon \Lambda^{I} \to \mathbb{R}$ map $y \mapsto \H(f(\bm u) \mid \bm u_{I} = y)$. Then there exists a set $S \subseteq \Lambda^{I}$ such that $\Pr_{\bm r \sim \Lambda^{I}}[\bm r \in S] \ge 1 - |I| \cdot \sqrt{\delta}$ and $h$ has $O(\mu \log (m |\Sigma|))$-bounded differences over $S$.
\end{corollary}
\begin{proof}
    Let us fix some $i \in I$ and show that there exists a set $S_i$ with $\Pr[\bm r \in S_i] \ge 1 - \sqrt{\delta}$ such that $h$ has bounded differences in the $i$-th coordinate in the set $S_i$. The claim then follows by the union bound.

    Let $S_i$ be set of $r \in \Lambda^{I}$ such that $f_{r_{I \setminus \{i\}}}$ is $(\mu,\sqrt{\delta})$-Lipschitz, by \cref{lem:Lipschitzness-after-conditioning} the forest $f|_{\bm r_{I \setminus \{i\}}}$ is $(\mu, \sqrt{\delta})$-Lipschitz with probability $1-\sqrt{\delta}$ over the choice of $\bm r \sim \Lambda^{I}$. Then $g_r \coloneqq f|_{r_{I \setminus \{i\}}}$ is also average-$(\mu + \sqrt{\delta}m)$-Lipschitz. Since $\sqrt{\delta}m = O(\mu)$, an application of \cref{lem:entropy-deviation} implies that for every $y \in \Lambda$
    \[ |\H(g_r(\bm y, \bm z) \mid \bm y = y) - \H(g_r(\bm y, \bm z))| \le O(\mu \log (m |\Sigma|)),\]
    where $\bm y \sim \Lambda^{\{i\}}$ and $\bm z \sim \Lambda^{[s] \setminus I}$. Then notice that $h(x,y) = \H(g_y(\bm y, \bm z) \mid \bm y = y)$, so the above shows that $h$ has $O(\mu \log (m |\Sigma|))$-bounded differences in $S_i$ for the $i$-th coordinate, as required. 
\end{proof}

\begin{proof}[Proof of \cref{lem:first-entropy-concentration}]
     First we reduce the number of the input cells in $I$: we cluster it and increase the alphabet, so each cluster becomes a single symbol. The challenge is to do it so that average-$\mu$-Lipschitzness does not suffer. Let $\bm \theta_i$ for $i \in [s]$ be the number of queries to the input cell $i$ the trees in $f$ collectively make on the input $\bm u \sim [n]^s$. Observe that $\Exp[\sum_{i \in [s]} \bm \theta_i] \le md$, since each tree makes at most $d$ queries on the given input. Let $I_1, \dots, I_\ell$ be the partition of $I$ such that $\sum_{i \in I_j} \Exp[\bm \theta_i] \le \mu$ for every $j \in [\ell]$ and $\ell$ is minimized. Then no two sets in $I_1, \dots, I_\ell$ can be united so the property is satisfied, hence $\sum_{i \in I_j} \Exp[\bm \theta_i] \ge \mu / 2$ for all $j \in [\ell]$ except perhaps one, thus $\ell \le 2md / \mu + 1 \le 3md/\mu$. Then we set $\Lambda \coloneqq [n]^{\max_{j \in [\ell]} |I_j|}$ and for every tree in $f$ replace each query to the input cell $i \in I$ with the query to the cluster $I_j \ni i$. This does not affect the depth of the trees and each new input cell is queried at most $\mu$ times in expectation. Thus, the obtained forest $f'\colon \Lambda^\ell \times [n]^{[s] \setminus I} \to [n]^m$ is average-$\mu$-Lipschitz depth-$d$ decision forest such that $f(\bm u) \equiv f'(\bm y, \bm z)$ for $\bm y \sim \Lambda^\ell$, $\bm z \sim [n]^{[s] \setminus I}$. 
    
    Let $h\colon \Lambda^{\ell} \to \mathbb{R}$ map $y\in \Lambda^\ell$ to $\H(f'(y, \bm z)) = \H(f(y', \bm u_{[s] \setminus I}))$, where $y' \in [n]^I$ is the unclustered version of $y$. Then $\Exp[h(\bm r)] = \H(f(\bm u) \mid \bm u_I)$. 
    
    By \cref{cor:entropy-bd-differences} we have that $h$ has $O(\mu \log(m |\Sigma|)$-bounded differences on a set $S \subseteq \Lambda^\ell$ with probability mass $1 - \ell \sqrt{\delta}$. 
    So by \cref{lem:mc-diarmid} we have 
    \begin{align*}
         \Pr_{\bm r \sim [n]^{I}}[|h(\bm r) - \Exp_{\bm y \sim \bm r}[h(\bm y)]| \ge \lambda + o(1)]
         &= \Pr_{\bm r \sim [n]^{I}}[|h(\bm r) - \Exp_{\bm y \sim \bm r}[h(\bm y)]| \ge \lambda + O(\delta\ell\mu \log(m|\Sigma|)) \cdot \ell \sqrt{\delta}]
         \\\text{(by \cref{lem:mc-diarmid}) } &\le
    2\exp\left(-\Omega\left(\frac{\lambda^2}{c^2 \ell}\right)\right) + 2\ell\sqrt{\delta}\\
    &\le \exp\left(-\Omega\left(\frac{\lambda^2}{\mu \log^2(m |\Sigma|) m d}\right)\right) + \exp(-m^{\Omega(1)}).
    \end{align*}
   Replacing $\lambda$ with $\lambda' \coloneqq \lambda + o(1)$ we get the claimed inequality.  
\end{proof}

\subsection{Enforcing average Lipschitzness: proof of \cref{lem:strong-avg-Lipschitz}}
\label{sec:fixing-to-get-lipschitzness}
In this section we show that exhaustively fixing input cells that violate average Lipschitzness eventually to random values eventually makes any bounded depth decision forest average-Lipschitz. In \cite[Claim~3.10]{BIL12} it is shown for average-$\sqrt{n}$-Lipschitzness and above, we show it for the values below $\sqrt{n}$.
We need the following classical fact:
\begin{lemma}[Expected stopping time, see e.g. \cite{KK18}]
\label{lem:stopping-time}
    Suppose $\{\bm x_i\}_{i \in \mathbb{Z}_{>0}}$ is a sequence of random variables. Define the stopping time to be
    \(\bm t \coloneqq \min\{t \in \mathbb{Z}_{>0} \mid \bm x_i \ge N\},\) and assume it is finite. Then whenever $\Exp[\bm x_i - \bm x_{i-1} \mid \bm x_{<i}, i \le \bm t] \ge \varepsilon$, we have $\Exp[\bm t] \le N / \varepsilon$.
\end{lemma}
We first show a weaker version of the lemma, we will then see how to easily boost it to get exponential success probability.
\begin{lemma}
\label{lem:getting-avg-Lipschitz}
    Suppose $f\colon [n]^s \to [n]^m$ is an arbitrary depth-$d$ decision forest. There exists a $nd/(\varepsilon \mu)$-depth decision tree $T$ querying symbols of $[n]^s$ such that for a random leaf $\bm \ell$ of $T$, $f|_{\bm \ell}$ is average-$\mu$-Lipschitz with probability $1-\varepsilon$.
\end{lemma}
\begin{proof}
    We define the tree by describing a random walk from its root to a leaf. Let $\bm a_1, \bm a_2, \dots \in [s]$ be the random variables describing what input cells are fixed at each step of the walk and $\bm b_1, \bm b_2, \dots \in [n]$ describe the values they are fixed to. All $\bm b_i$ are independent and uniform over $[n]$.  Let $f_{\bm a_{<i}, \bm b_{<i}}$ denote the forest with the input cell $\bm a_j$ is fixed to $\bm b_j$ for all $j < i$. Then $\bm a_i$ is defined as an element of $[n] \setminus \{\bm a_1, \dots, \bm a_{i-1}\}$ such that the expectation of the number of queries to it by $f_{\bm a_{<i}, \bm b_{<i}}$ is the largest (breaking ties arbitrarily). If that value is less than $\mu$ the process stops, so in that case $\bm t = i-1$. 

    Let $\bm p_i$ be the expected number of queries made by $f_{\bm a_{\le i}, \bm b_{\le i}}$ on a uniform random input $\bm u \sim [n]^{[s] \setminus \{\bm a_1, \dots, \bm a_i\}}$. We claim that $\Exp[\bm p_{i-1} - \bm p_{i} \mid \bm a_{<i}, \bm b_{<i}] \ge \mu$. Let $\bm q_{i, S}$ be the expected number of queries to $S \subseteq [s]$ made by $f_{\bm a_{\le i}, \bm b_{\le i}}$ on a uniform random input, so $\bm p_i = \bm q_{i, [s]}$. Then
    \begin{align*}
         \Exp[\bm p_{i-1} - \bm p_{i} \mid \bm a_{<i}, \bm b_{<i}] & = \bm p_{i-1} - \Exp[\bm q_{i, [s] \setminus \bm a_i} + \bm q_{i, \bm a_i} \mid \bm a_{<i}, \bm b_{<i}]\\
         \text{($\bm a_i$ is already assigned) } &= \bm p_{i-1}  - \Exp[\bm q_{i, [s] \setminus \bm a_i} \mid \bm a_{<i}, \bm b_{<i}] \\
         &= \Exp[\bm q_{i-1, [s] \setminus \bm a_i} - \bm q_{i, [s] \setminus \bm a_i} + \bm q_{i-1, \bm a_i} \mid \bm a_{<i}, \bm b_{<i}]\\
         \text{(expected number of queries decreases) }&\ge \Exp[\bm q_{i-1, \bm a_i} \mid \bm a_{<i}, \bm b_{<i}]\\
         &\ge \mu.
    \end{align*}
    Applying \cref{lem:stopping-time} to $\bm p_0 - \bm p_i$ we get that $\Exp[\bm t] \le \bm p_0 / \mu \le nd/\mu$. Hence, the expected depth of a leaf of the tree is at most $nd/\mu$. Let us then prune all branches of the tree at depth more than $nd/(\varepsilon \mu)$ and get observe that we cut just $\varepsilon$-fraction of the leaves by Markov's inequality on $\bm t$. 
\end{proof}
We are now ready to prove the stronger version.
\strongavglipschitz*
\begin{proof}
    Applying \cref{lem:getting-avg-Lipschitz} with $\varepsilon=1/2$ we get that there exists a tree of depth $2nd/\mu$-depth decision tree $T_0$ such that for its random leaf $\bm \ell$ we have that $f|_{\bm\ell}$ is average-$\mu$-Lipschitz with probability $1/2$. We say that a leaf $\ell$ is \emph{successful} if $f|_\ell$ is average-$\mu$-Lipschitz, otherwise it is \emph{failed}. Let us construct the tree $T_1$ by taking $T_0$ and for its every leaf $\ell$ such that $f|_{\ell}$ is \emph{not} average-$\mu$-Lipschitz hang a tree $T'$ obtained by applying \cref{lem:getting-avg-Lipschitz} to $f|_{\ell}$ with $\varepsilon=1/2$. Define $T_2, \dots, T_{\log (1/\varepsilon)}$ the same way: $T_i$ is obtained from $T_{i-1}$ by hanging trees given by \cref{lem:getting-avg-Lipschitz} to all its failed leaves.

    Then consider a random walk down $T_{\log (1/\varepsilon)}$: with probability $1/2$ it ends in a successful leaf of $T_0$, conditioned on it passing through a failed leaf of $T_0$ with probability $1/2$ it ends in a successful leaf of $T_1$. Hence with probability $1 - 2^{-\log(1/\varepsilon)} = 1-\varepsilon$ the walk down $T_{\log(1/\varepsilon)}$ terminates in a successful leaf.
\end{proof}

\section{Containment Lemma}
\label{sec:containment}
In this section we prove a lemma that formalizes the intuition that if the sampled distribution has low entropy it must be very far from the target high-entropy distribution. While this is always true in a weak sense (see \cref{lem:simple-containment}), for Lipschitz functions the distance can be boosted to exponentially close to $1$.
\containmentlemma*

Our starting point is \cref{lem:simple-containment}: since $\H(f(\bm u)) \ge m \log n / 4$, the lemma implies that there exists an event $G \subseteq [n]^m$ of size $n^{m/2}$ such that $\Pr[f(\bm u) \in G] \ge 1/2$. In other words, at least half of the elements in $[n]^s$ are mapped to a set of outputs of size at most $n^{m/2}$. Our goal is to enlarge this preimage so that it covers almost the entire $[n]^s$, while keeping the image size relatively small, as shown in \cref{fig:low-entropy-error-boosting}. 
To proceed further, we will need the following lemma, which we will prove later:
\begin{lemma}
    \label{lem:coupling}
    Let $T\colon \Sigma^M \to \{0, 1\}$ be a depth-$k$ decision tree with
    $\mu \coloneqq \Pr_{\bm x\sim \Sigma^M}[T(\bm x)=1]$. Then, there exists $2$-dimensional
    distribution $\mathcal{D}$ with marginals uniform over $\Sigma^M$ and $T^{-1}(1)$ respectively such that
    \[
        \Exp_{(\bm x, \bm y)\sim\mathcal{D}}[\dist(\bm x, \bm y)] \leq O\left(\sqrt{k\log(1/\mu)}\right).
    \]
\end{lemma}
The set $f^{-1}(G)$ can be recognized by a decision tree $T'\colon [n]^s \to \{0,1\}$ of depth $nd$: $T(x)$ just computes $f_i(x)$ one by one and then accepts iff the resulting vector $(f_1(x), \dots, f_n(x))$ is in $G$. Then we are in a position to apply \cref{lem:coupling} with $\mu \coloneqq 1/2$: let $(\bm u, \bm y)$ be coupled with $\bm u \sim [n]^s$, $\bm y \sim f^{-1}(G)$ and
\begin{equation}
    r \coloneqq O(\sqrt{nd})\geq \Exp\left[\dist(\bm u, \bm y)\right] \geq
    \Exp\left[\min_{ y \in f^{-1}(G)}\dist(\bm u,  y)\right] = \Exp[\dist(\bm u, f^{-1}(G))], \label{eq:radius-def}
\end{equation}
At this point we use Lipschitzness. First, we apply \cref{lem:average-Lipschitz-to-Lipschitz} to get that $f$ is in fact $(n^{0.2}, \delta)$-Lipschitz for $\delta = \exp(-n^{\Omega(1)})$. 
We now claim that we can use Lipschitzness to say that $f(\bm u)$ is close to $G$ in expectation.
\begin{claim} \label{claim:lipschitz-distance}
    $\Exp[\dist(f(\bm u), G)] = \Exp[\min_{o \in G} \dist(f(\bm u), o)] \le O(n^{0.8})$.
\end{claim}
\begin{proof}
      Let $E \subseteq [n]^s$ be the set of inputs where $(n^{0.2}, \delta)$-Lipschitzness is violated, i.e. some input cell $i \in [s]$ is queried by more than $n^{0.2}$ trees of $f$. We then have $\Pr_{\bm u \sim [n]^s}[\bm u \in E] \le s \delta = \exp(-n^{\Omega(1)})$. Consider any $x \not\in E$ and $y \in [n]^s$. Then on the input $x$ the set of input cells $\{i \in [s] \mid x_i \neq y_i\}$ is queried by at most $n^{0.2} \cdot \dist(x,y)$ trees in $f$. Let $D_x \subseteq [m]$ be the set of these trees. The output of the trees outside $D_x$ does not change when we replace the input $x$ with $y$: $f_{[n] \setminus D_x}(x) = f_{[n] \setminus D_x}(y)$. Consequently $\dist(f(x), f(y)) \le n^{0.2} \cdot \dist(x,y)$. Thus, we write
\[\Exp\left[  \dist( f(\bm u), G) \right] \leq \Exp\left[n^{0.2} \cdot \dist(\bm u, f^{-1}(G)) \,\middle|\, \bm u \not\in E\right] + s \cdot \Pr[\bm u \in E] = O(n^{0.8}). \qedhere\]
\end{proof}
Let $\mu = n^{0.1}$ be the average Lipschitzness parameter of $f$. Following the same strategy of clustering coordinates as in \cref{lem:first-entropy-concentration}, we get a partition $I_1, I_2, \ldots, I_\ell$ of $[s]$ with property 
$$\sum_{i\in I_j}\Exp[\bm \theta_j] \leq 2\mu$$
for every $j\in [\ell]$ ($\bm \theta_j$ here, as before, denotes number of trees querying $j$ on input $\bm u$), with $\ell \leq 3 md/\mu$. Let $f'\colon [n]^{I_1}\times[n]^{I_2}\times\ldots\times[n]^{I_\ell}\to[n]^m$ be the clustered version of $f$. By construction, it is average-$2\mu$-Lipschitz, so we apply \cref{lem:average-Lipschitz-to-Lipschitz} (setting the input alphabet $\Lambda \coloneqq [n]^{\max_{i \in [\ell]} |I_i|}$) and find that $f'$ is $(6\mu d^2\log(1/\delta), \delta)$-Lipschitz, i.e. $(n^{0.2}, \delta)$-Lipschitz for $\delta = \exp(-n^{\Omega(1)})$. We will abuse the notation and write $f'(x)$ for $x\in [n]^s$ meaning $f'(x_{I_1}, x_{I_2}, \ldots, x_{I_\ell})$, and identify $[n]^s$ with $[n]^{I_1}\times[n]^{I_2}\times\ldots\times[n]^{I_\ell}$.

Let us define the function $h\colon [n]^{I_1}\times[n]^{I_2}\times\ldots\times[n]^{I_\ell}\to\mathbb{R}$ as follows:
\[h(x) \coloneqq \dist(f'(x), G) = \min_{o\in G}\dist(f'(x), o).\]

Let $S \subseteq [n]^s$ be the set of inputs, such that every position of $f'$ is queried at most $n^{0.2}$ times. By the definition of Lipschitzness, $\Pr_{\bm u\sim [n]^s}[\bm u\in S] \geq 1-s\delta$. Analogously to \cref{claim:lipschitz-distance}, $h$ has $n^{0.2}$-bounded differences over $S$, so we apply \cref{lem:mc-diarmid} and get
\[
\Pr_{\bm u\sim [n]^s}\Bigl[
  |h(\bm u) - \Exp_{\bm u\sim [n]^s}[h(\bm u)\mid \bm u \in S]| 
  \geq \lambda + s\delta \cdot n^{0.2}\cdot \tfrac{2md}{\mu}
\Bigr] 
\leq 2s\delta + 2\exp\!\left(
  -\Omega\!\left(
    \frac{\lambda^2}{n^{0.4}\cdot md/\mu}
  \right)\right).
\]

So, substituting $\lambda = n^{0.8}$ and noting that \(\left|\Exp[h(\bm u)\mid \bm u\in S] - \Exp\left[ \dist(f(\bm u), G) \right]\right| \leq s\delta m,\)
we get 
\(\Pr\left[ f(\bm u), G) \geq Cn^{0.8}\right]\leq \exp(-n^{\Omega(1)})\)
for a large enough constant $C > 0$.

Finally, we define $F \coloneqq \mathcal{N}_{Cn^{0.8}}(G)$, where $\mathcal{N}_r(P) \coloneqq \{x\mid\exists y\in P~\dist(x, y)\leq r\}$. We can upper bound the size of $F$ as
$|F|\leq |G|\cdot{n\choose Cn^{0.8}}\cdot n^{Cn^{0.8}}\leq n^{3m/4}$
and for such $F$, we have $\Pr[f(\bm u)\in F] \geq 1-\exp(-n^{\Omega(1)}).$

\subsection{Proof of \cref{lem:coupling}}
\renewcommand{\path}{\mathsf{path}}

Recall that the statistical (total variation) distance between two \emph{distributions} $\nu_1$ and $\nu_2$ can be defined through optimal couplings:
\[ \Delta(\nu_1, \nu_2) = \min\left\{\Pr_{\bm a,\bm b \sim \mathcal{C}}[\bm a \neq \bm b] \,\middle|\, \mathcal{C}\colon\text{distribution with marginals $\nu_1$, $\nu_2$}\right\}.   \]
We say that $\mathcal{C}$ that achieves this minimum\footnote{We work with variables over a finite support, so the minimum always exists.} is \emph{the optimal coupling} of $\nu_1$ and $\nu_2$.

The coupling $\mathcal{D}$ claimed by \cref{lem:coupling} is basically a composition of couplings $\mathcal{C}$ for each pair of symbols in $\bm x$ and $\bm y$: for each node of $T$ querying $i$ we enforce that the distributions of $\bm x_i$ and $\bm y_i$ conditioned on both $\bm x$ and $\bm y$ passing through the node are optimally coupled, see \cref{alg:coupling}.
\medskip

    Wlog we assume that $T$ is a full depth-$k$ decision tree, and at any path, it does not query the same input cell twice.
    Let $\bm x \sim \Sigma^M$. We denote by $\path(r)$ the computation path of $r$ in $T$. We construct $\bm y \in T^{-1}(1)$ using \cref{alg:coupling}.
    
    \begin{algorithm}
    \caption{\small The algorithm defining the coupling $\mathcal{D}$.}
    \label{alg:coupling}
    \begin{algorithmic}[1]
    \State $\bm y \gets \bm x$.
    \State $v \gets$ root of $T$.
    \While{$v$ is not a leaf}
      \State Let $i \in [M]$ be the coordinate queried by $v$, and $\{w_i\}_{i \in \Sigma}$ be its children.
      \State Let $\mathcal{C}$ be the optimal coupling of $\bm x_i$ and $(\bm z_i \mid \path(\bm z) \ni v)$ where $\bm z \sim T^{-1}(1)$. 
      \State Define $(\bm y_i \mid \path(\bm y) \ni v)$ such that $(\bm x_i, (\bm y_i \mid \path(\bm y) \ni v))$ is distributed according to $\mathcal{C}$.
      \State $v \gets w_{\bm y_i}$.
    \EndWhile
    \end{algorithmic}
\end{algorithm}

Let $\bm z \sim T^{-1}(1)$. By definition for every internal node $v$ of $T$ that queries $i$, the distributions $(\bm y_i \mid \path(\bm y) \ni v)$ and $(\bm z_i \mid \path(\bm z) \ni v)$ coincide.
Thus, the next branch taken from $v$ on $\path(\bm y)$ has the same conditional distribution as for $\bm z$; by backward induction over the depth, $\path(\bm y)$ is uniform over accepting paths. The coordinates not queried by $T$ remain uniform ($\bm y_j = \bm x_j$), therefore $\bm y$ is uniform over $T^{-1}(1)$. 

Now it remains to bound $\mathrm{dist}(\bm x, \bm y)$. Let $\bm p =(\bm p_1,\ldots,\bm p_k) \in \Sigma^k$ encode a uniformly random \emph{accepting} root-to-leaf path in $T$, and for $\tau \in \{0,1\}^j$ let $v(\tau)$ denote the node in $T$ that is reached by going from the root according to $\tau$. At a node $v$ querying $i$ let 
\[\delta_v \coloneqq \Pr[\bm y_i \neq \bm x_i \mid \path(\bm x) \ni v\, \land\, \path(\bm y) \ni v] = \Delta(\bm x_i, (\bm y_i \mid \path(\bm y) \ni v)).\] Then
\[
\Exp[\mathrm{dist}(\bm x, \bm y)] \leq \sum_{j \in [k]} \Exp[\delta_{v(\bm p_{\le j})}].
\]

Pinsker's inequality (see \cite[Lemma~11.6.1]{InfTheoryBook}) 
applied to $(\bm y_i \mid \path(\bm y) \ni v)$ and the uniform distribution $\bm u \sim \Sigma$ (i.e. the distribution of $\bm x_i$) states
\[ \delta_v^2 = \Delta((\bm y_i \mid \path(\bm y) \ni v), \bm u)^2 \le {2 \ln 2} (\H(\bm u) - \H(\bm y_i \mid \path(\bm y) \ni v)), \] 
so $\log |\Sigma| - \delta_v^2 / 2 \ln 2 \ge \H(\bm y_i \mid \path(\bm y) \ni v)$. The event ``$\path(\bm y) \ni v$'' is equivalent to ``$\bm p_{<j} = \alpha$'', and $\bm p_j$ is determined by $\bm y_i$. Hence, we have
\[\H(\bm p_j\mid \bm p_{<j}=\alpha) \leq \log |\Sigma|-\delta_{v(\alpha)}^2/2\ln2. \]

Taking the expectation over $\alpha$, and summing over $j$, we obtain
\[
k\log|\Sigma|-\Big(\sum_{j\in [k]}\Exp[\delta_{v(\bm p_{\le j})}^2]\Big)/2\ln 2 \geq \sum_{j\in [k]} \H(\bm p_{j} \mid \bm p_{<j})=\H(\bm p) = k \log |\Sigma| - \log(1/\mu).
\]

By Jensen's and Cauchy--Schwarz's inequalities we get,
\[
    2\ln 2\cdot\log(1/\mu) \geq \sum_{j\in [k]}\Exp[\delta_{v(\bm p_{\le j})}^2] \geq
    \sum_{j \in [k]}\left(\Exp[\delta_{v(\bm p_{\le j})}]\right)^2\geq \frac 1k \Big(\sum_{j\in [k]}
      \Exp[\delta_{v(\bm p_{\le j})}]\Big)^2
\]

Therefore,
$$
\Exp[\mathrm{dist}(\bm x,\bm y)]\leq \sum_{j\in [k]}\Exp[\delta_{v(\bm p_{\le j})}] \leq O\left(\sqrt{k\log(1/\mu)}\right).
$$

\section{Collision Lemma}
\label{sec:proof-of-collision-lemma}

In this section, we prove the general version of our collision lemma. We first prove the version for the independent random variables (that directly generalizes \cref{lem:simplified-collision-lemma}) and then will derive \cref{lem:collision-for-d-1}, which is a simple corollary of the independent variables case.
\begin{lemma}
    \label{lem:high-entropy-independent}
    Let $\bm z_1, \dots, \bm z_m$ be independent random variables supported over $[n] \cup
    \{\bot\}$. Assume that for every $i$ we have $\H(\bm z_i) \ge \delta \cdot m \log n$ and $m \ge
    n^{1-\varepsilon}$ and $\delta = \delta(n) \geq \max(2\log\log n/\log n, 4\varepsilon)$. Then
    \[
        \Pr[\exists i \neq j \in [m] \colon \bm z_i = \bm z_j \neq \bot] \ge
        1-\exp(-\Omega(\delta^4m^3/n^2)).
    \]
\end{lemma}

\subsection{Proof of \cref{lem:high-entropy-independent}}
Our goal is to show that with high probability there exists a collision among the $\bm z_i$’s. 
We visualize the setting with a complete bipartite graph $H=([m], [n] \cup \{\bot\}, [m] \times ([n] \cup \{\bot\}))$ in \cref{fig:collision}, the upper part 
nodes correspond to the variables $\bm z_1, \dots, \bm z_m$, the lower part nodes correspond to their values
in $[n] \cup \{\bot\}$, every edge $(i, k)$ is labeled with the probability $p_k^i \coloneqq \Pr[\bm z_i = k]$. 

Then every value $z_1, \dots, z_m \in ([n] \cup \{\bot\})^m$ corresponds to a set of edges $\{(i, z_i) \mid i \in [m]\}$
in $H$. We then have a collision 
iff there is a node $k \in [n]$ with a degree at least $2$ in the edges $\{(i, z_i) \mid i \in [m]\}$.  We define the function $f$ as a smoothed version of counting the number of nodes in $[m]$ of degree at least $2$.
\begin{equation}
    \label{eq:f-definition}
    f(z_1, \dots, z_m) \coloneqq \sum_{k \in [n]} \max \bigl(0, \bigl| \{i \in [m] \mid z_i = k\} \bigr| -
    1 \bigr),
\end{equation}
Then by definition there is a collision if and only if the value of $f$ is non-zero:
\[
    \Pr[\exists i \neq j\colon \bm z_i = \bm z_j \neq \bot] =
    \Pr[f(\bm z_1,\dots,\bm z_m) \neq 0],
\] 
thus it suffices to show that $f(\bm z_1, \ldots, \bm z_m)$ is nonzero with high probability.
The key feature of thus defined $f$ is that it satisfies the $2$-bounded differences property 
over its entire domain: indeed if $z$ and $z'$ differ in a coordinate $i \in [m]$ then the 
only summands in \eqref{eq:f-definition} that may change are $k \in \{z_i, z'_i\}$, each
of those can change by at most by one.
Therefore \nameref{lem:mc-diarmid} applies:
\[ \Pr[f(\bm z) = 0] 
\le \Pr[f(\bm z) \le \Exp[f(\bm z)] / 2] 
= \exp(-\Omega(\Exp[f(\bm z)]^2 / m)).\]
Thus, the main technical part of the proof is to show that $\Exp[f(\bm z)] = \Omega(\delta^2 m^2/n)$.

\paragraph{Bounding the expectation.} 
By the definition of $f$ we have \(\Exp[f(\bm z)] 
    \ge \sum_{k \in [n]} \Pr[\exists i \neq j \in [m]\colon \bm z_i = \bm z_j = k].\) 
For each $k \in [n]$ the probability that there exist $\bm z_i = \bm z_j = k$ can be computed explicitly: indeed we have independent events ``$\bm z_i = k$'' and the probability we bound is that at least two of these events occur. If all $\Pr[\bm z_i = k]$ are small enough, we can use the following simple bound:
\begin{claim}
\label{cl:at-least-two}
    Suppose events $E_1, \dots, E_\ell$ are independent, each $E_i$ occurs with probability $q_i \le \alpha$ such that $\bar{q} \coloneqq \sum_{i\in[\ell]} q_\ell \le 1/8$. Then
    \[ \Pr\Big[\sum_{i \in [\ell]} \lb E_i \rb \ge 2 \Big] \ge \bar{q}^2 / 4 - 2 \alpha \bar{q}. \]
\end{claim}

Let $\bar{p}_k \coloneqq \sum_{i \in [n]} p_k^i$. If \cref{cl:at-least-two} applied for every $k \in [n]$ with a negligible $\alpha$, we would get that $\Exp[f(\bm z)]$ is at least $(1-o(1)) \cdot \sum_{k \in [n]} \bar{p}_k^2$. Since $\sum_{k \in [n]} \bar{p}_k \approx m$, we could apply Cauchy-Shwartz's inequality to get the desired bound. We face two technical problems: $\bar{p}_k$ might be too large, and $p_k^i$ (and thus $\alpha$) may not be negligible. We fix both of those directly: remove too large $p_k^i$ and split the node $k$ in the graph $H$ into several nodes so that the $\bar{p}$-value for each of the nodes does not exceed $1/8$. 

\begin{figure}[h]
    \centering
    \tikzset{
    rectangle-set/.style = {
        #1,
        pattern = {Lines[angle = 45, distance = 1pt, line width = 0.2pt]},
        pattern color = #1,
        thick,
        rounded corners = 3pt
    }
}
\begin{tikzpicture}
    \def\a{4.5}
    \def\b{5}
    \def\d{3}

    \draw[thin, densely dashed, rounded corners = 1pt] (-\a, \d) rectangle (\a, \d + 0.3);
    \draw[thick] (-\b, 0.15) -- ++(0, -0.15) -- ++(2 * \b, 0) -- ++(0, 0.15);
    \draw[thick] (\b - 0.5, 0) -- ++(0, 0.15);
    
    \node at (\b - 0.25, -0.3) {$\bot$};
    \node at (-\b - 0.4, 0.4) {$[n]$};

    \node at (\a + 0.5, \d + 0.65) {$[n]^m$};

    \foreach \x [count = \i] in {-\a + 0.5, -1, 1, \a - 0.5}{
        \node[in-out-node = {ForestGreen}] (a\i) at (\x, \d + 0.15) {};
    }

    \foreach \x [count = \i] in {0.7, 1.2, ..., 2.2, \b - 1.5, \b + 1, 2 * \b - 0.75, 2 * \b - 0.25}{
        \node[in-out-node = {ForestGreen}] (b\i) at (-\b + \x, 0.15) {};
    }

    \node at (1, -0.3) {$k$};

    \draw[thick, VioletRed] (-1.4, 0.4 + \d) -- ++(-0.1, 0) -- ++(0, -0.5) -- ++(0.1, 0);
    \draw[thick, VioletRed] (1.4, 0.4 + \d) -- ++(0.1, 0) -- ++(0, -0.5) -- ++(-0.1, 0);
    \node at (0, 1 + \d) {$G_k$};
    \node at (-1, 0.6 + \d) {$\bm z_i$};
    \node at (1, 0.6 + \d) {$\bm z_j$};
    \node at (-\a + 0.5, 0.6 + \d) {$\bm z_1$};
    \node at (\a - 0.5, 0.6 + \d) {$\bm z_m$};

    \foreach \i in {7, 8}{
        \draw[ForestGreen] (a3) -- (b\i);
    }
    \foreach \i in {1, 2, ..., 5}{
        \draw[ForestGreen] (a2) -- (b\i);
    }
    \draw[
        decorate,
        decoration = {brace, amplitude = 3pt, mirror}]
        (-\b + 0.6, -0.15) -- ++(1.7, 0)
        node[midway, below] {\tiny $n^{-\delta / 2}$ heavy supp of $\bm z_i$};

    \draw[VioletRed, thick,
         decorate, decoration = {snake, segment length = 6mm, amplitude = 0.3mm}]
         (a2) -- (b6) node [midway, below] {\tiny $p_k^i$};
    \draw[VioletRed, thick,
        decorate, decoration = {snake, segment length = 6mm, amplitude = 0.3mm}]             
        (a3) -- (b6)  node [midway, right] {\tiny $p_k^j$};

    \foreach \i in {1, 4}{
        \foreach \j in {-1, 0, 1}{
            \draw[dashed, ForestGreen, path fading = south]
                ([xshift = 1pt, yshift = 0pt]a\i) -- (a\i) -- ++(-90 + 20 * \j: 1.4);
        }
    }

    \node at (-\a - 2, 0) {};
    \node at (\a + 2, 0) {};
\end{tikzpicture}
    \caption{\small This picture illustrates the approach to prove that the probability of having \emph{collision value} $k$ is significant, i.e. $p \coloneqq \Pr[\exists i\neq j \in [m]\colon \bm z_i = \bm z_j = k]$ is bounded away from zero. The key technical step in the proof is to remove ``heavy'' edges from the graph---the ones with $p_k^i > n^{-\delta/2}$. $G_k$ denotes the neighborhood of $k$ in $H$ with all heavy edges removed.}
    \label{fig:collision}
\end{figure}
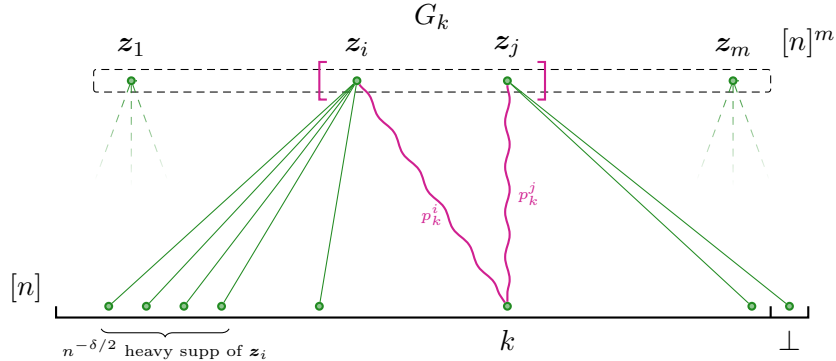

In the first step we remove from $H$ all edges with too large $p_k^i$, the following claim implies that we retain significant probability mass after this removal:
\begin{claim}
    \label{claim:high-entropy-var-new}
    Let $\bm a$ be a random variable over $[n]$ with $\H(\bm a) \ge c\log n$, where $c = c(n) >
    4 / n$. Then with $p\colon [n] \to [0, 1]$ being the probability function of $\bm a$, we have
    \[
        \Pr[p(\bm a) \le n^{-c / 2}] = \sum_{p(i) \leq n^{-c / 2}} p(i) \ge c / 8.
    \]
\end{claim}

Applying \cref{claim:high-entropy-var-new} with $\bm a = \bm z_i$, $c = \delta$ we get that $\sum_{(i,k) \in [m] \times ([n] \cup \{\bot\})} p_k^i \ge \delta m / 8$. Observe that we can now remove the node $\bot$ and still retain most of the probability mass:
\[ \sum_{i \in [m]\colon p_\bot^i \le n^{-\delta/2}} p_\bot^i \le m n^{-\delta /2} \le m / \log n \le \delta m / 16. \]
The last two inequalities utilize $\delta \ge 2\log\log n / \log n$. 
Thus, we get
\[
    \sum_{(i, k) \in [m] \times [n] \colon p_k^i \leq n^{-\delta / 2}}p_k^i \geq \Omega(\delta m).
\]

Let $G \subseteq [m] \times [n]$ be the set of the light edges of $H$ defined as $G\coloneqq\{(i,k) \mid p^i_k
\le n^{-\delta/2}\}$ and for every $k \in [n]$ let $G_k \coloneqq \{i \in [m] \mid (i,k) \in G\}$ and
$\bar{p}_k \coloneqq \sum_{i \in G_k} p_k^i$. In order to force $\bar{p}_k \le 1/8$, we
split $G_k$ into the smallest number of parts $G_k^1, \ldots G_k^{t_k}$ with the property
\[
    \forall h\in [t_k]\ \sum_{i\in G_k^h}p_k^i \leq 1/8.
\] 
It is possible, since $n^{-\delta/2} \leq 1/8$, and we will get $t_k\leq 16\bar{p}_k + 1$, so the total
number of parts will be $O(n)$. Clearly,
\[
  \Exp[f(\bm z_1, \ldots, \bm z_m)]
     \geq \sum_{k\in [n]}\sum_{h\in[t_k]}\Pr[\exists i\neq j\in G_k^h\colon \bm z_i=\bm z_j=k].
\]
Now we denote $\bar{p}_{kh} = \sum_{i\in G_k^h}p_k^i$ and get 
by \cref{cl:at-least-two} applied to the events ``$\bm z_i = k$'' for $i \in G_k^h$ that:
\[ \Pr[\exists i \neq j \in G_k^h\colon \bm z_i = \bm z_j = k] \ge \bar{p}_{kh}^2 / 4 - 2 n^{-\delta/2} \bar{p}_{kh}.\]

Finally, we obtain
\begin{align*}
  \Exp[f(\bm z_1, \ldots, \bm z_m)]
    &\geq \frac{1}{4}\sum_{k\in [n]}\sum_{h\in[t_k]} \bar{p}_{kh}^2 -
      2n^{-\delta/2}\sum_{k\in [n]}\sum_{h\in[t_k]} \bar{p}_{kh} \\
  \text{(by Cauchy--Schwarz's inequality) }
    &\geq \left(\sum_{k\in [n]}\sum_{h\in[t_k]} \bar{p}_{kh}\right)^2 / O(n) - 2n^{-\delta/2}
      \sum_{k \in [n]}\sum_{h\in[t_k]} \bar{p}_{kh} \\
  \text{(using $\delta \geq 4\varepsilon$) }
    &\geq \frac12\left(\sum_{k\in [n]}\sum_{h\in[t_k]} \bar{p}_{kh}\right)^2 / O(n) \geq
      \Omega(\delta^2m^2/n).
\end{align*}

\subsubsection{Proof of \cref{cl:at-least-two}}
Let $p \coloneqq \Pr[\sum_{i \in [\ell]} \lb E_i \rb \ge 2]$. The event ``$\sum_{i \in [\ell]} \lb E_i \rb \ge 2$'' does not occur iff none of $E_i$ occur, or exactly one of them occurs. Thus
\[ p = 1 - \prod_{i \in [\ell]} (1-q_i) - \sum_{i \in [\ell]} q_i / (1-q_i) \cdot \prod_{i \in [\ell]} (1-q_i).\]
We then rewrite using $q_i \le \alpha$:
\( 1 - p \le \prod_{i \in [\ell]} (1-q_i) \cdot (1 + \bar{q} / (1-\alpha)).\)
Since the geometric mean does not exceed the arithmetic mean, we have $\prod_{i \in [\ell]} (1-q_i) \le (\sum_{i \in [\ell]} (1-q_i) / \ell)^\ell = (1-\bar{q}/\ell)^\ell$. Moreover, since $\alpha \le \bar{q} < 1/2$ we have $1/(1-\alpha) \le 1 + 2\alpha$. Then
\[ 1 - p \le (1-\bar{q}/\ell)^\ell (1 + \bar{q} (1 + 2 \alpha)).\]
Now we use a simple analytical fact to bound the first multiplier:
\begin{fact}
\label{fact:taylor}
    For $x \in (0,1)$ and $n \ge 2$ the inequality $(1-x)^n \le 1 - nx + (nx)^2 / 2$ holds.
\end{fact}
\begin{proof}
    By Taylor's theorem \cite[Theorem~5.15]{Rud76} for $(1-x)^n$ at $x_0=0$ we get that there exists $\xi \in (0,x)$ such that $(1-x)^n = 1 - nx + n(n-1)/2 \cdot (1-\xi)^{n-2} x^2 \le 1 - nx + (nx)^2 / 2.$
\end{proof}
Then by \cref{fact:taylor} we conclude:
\begin{align*} 1 - p &\le (1 - \bar{q} + \bar{q}^2/2)(1 + \bar{q}(1+2\alpha))\\
&= 1 - \bar{q}^2/2 + 2\alpha \bar{q} + \bar{q}^3 \alpha + \bar{q}^3 / 2 - 2\alpha \bar{q}^2\\
\text{(since $\bar{q} \le 1/8$) }&\le 1 - \bar{q}^2 / 4 + 2\alpha \bar{q}.
\end{align*}

\subsubsection{Proof of \cref{claim:high-entropy-var-new}}
    $$\sum_{p(i) > n^{-c/2}} -p(i)\log p(i) \leq \sum_{p(i) > n^{-c/2}} \frac{c}{2}\cdot\log n\cdot p(i)\leq c\log n/2, $$
    Function $x\mapsto-x\log x$ is ascending at $[0, 1/e]$, so
    $$\sum_{p(i) < n^{-2}}-p(i)\log p(i) \leq n\cdot n^{-2}\cdot \log n=\log n/ n \leq c \log n /4.$$
    Combining the above inequalities with $\H(\bm a) \geq c\log n$, we get
    $$\sum_{n^{-2} \leq p(i) \leq n^{-c/2}}p(i) \geq \sum_{n^{-2} \leq p(i)\leq n^{-c/2}} -\frac{p(i)\log p(i)}{2\log n} \geq \frac{c\log n/4}{2\log n} = c/8$$

\subsection{Depth-$1$ decision forests}\label{subsec:depth-1}
In this section we prove a natural corollary of \cref{lem:high-entropy-independent}: virtually the same bound holds for functions that are computable with depth-$1$ decision forests.

\depthonecollision*
\begin{proof}
In order to apply \cref{lem:high-entropy-independent} we need to choose a subset $I \subseteq [m]$ such that the output cells in $I$ are independent and the entropy rate is not reduced too severely. All trees in $f$ can be partitioned into subsets $J_1 \sqcup \dots \sqcup J_\ell = [m]$ where in each $J_i$ the trees query the same input cell, which is different for different sets. Observe that $\H(f_{J_i}(\bm u)) \le \log n$. Thus, we have that $\ell \ge \H(f(\bm u)) / \log n \ge \delta m$.

We first form $I'$ by taking a representative $j_i \in J_i$ that maximizes $\H(f_{j_i}(\bm u))$ for each $i \in [\ell]$. We use the following simple fact
\begin{fact}
\label{fact:frac-inequality}
    For $a_1, \dots, a_n, b_1, \dots, b_n \in \mathbb{R}_{\ge 0}$ we have $\sum_{i \in [n]} (a_i / b_i) \ge (\sum_{i \in [n]} a_i)^2/(\sum_{i\in [n]} a_i b_i)$.
\end{fact}
\begin{proof}
    We first rewrite 
    \[\sum_{i \in [n]} \frac{a_i}{b_i} = \Big(\sum_{i \in [n]} a_i\Big) \cdot \sum_{i \in [n]} \frac{a_i}{\sum_{i \in [n]} a_i} \cdot \frac{1}{b_i}.\]
    Then applying Jensen inequality for the function $1/x$ we get that the right multiplier is at least $(\sum_{i \in [n]} a_i)/(\sum_{i \in [n]} a_i b_i)$, which concludes the proof.
\end{proof}

Then we get
\begin{multline*}
    \H(f_{I'}(\bm u)) 
    = \sum_{i \in [\ell]} \H(f_{j_i}(\bm u))
    \ge \sum_{i \in [\ell]} \frac{\H(f_{J_i}(\bm u))}{|J_i|}
   \overset{\text{\cref{fact:frac-inequality}}}{\ge} \\\frac{(\H(f(\bm u)))^2}{\sum_{i \in [\ell]} |J_i| \H(f_{J_i}(\bm u))}
    \ge \frac{(\delta m \log n)^2}{\log n \cdot m}
    = \delta^2 m \log n \ge \delta \ell \log n.
\end{multline*}
Now we pick $I \subseteq I'$ of $i \in I'$ such that $f_i(\bm u) \ge \delta \log n / 2$. Then $\delta \ell \log n \le \H(f_{I'}(\bm u)) \le |I| \log n + (\ell - |I|) \delta \log n / 2$, so $|I| \ge \delta \ell (1-\delta / 2) / 2 \ge \delta \ell / 4$. Then we apply \cref{lem:high-entropy-independent} to $f_I$, so $m' = (\delta^2 / 4) \cdot m$, and $\H(f_i(\bm u)) \ge (\delta/2) \log n$, so we need $\delta / 2 \ge \max(2\log \log n / \log n, 4\epsilon)$, which is what we have by the assumption. 
\end{proof}

\subsection{Simplified collision lemmas}
In this section, we derive \cref{lem:simplified-collision-lemma,lem:simplified-collision-lemma2}.

\simpcollisiona*
\begin{proof}
    We apply \cref{lem:collision-for-d-1} with $\delta = 1/8$, and $\epsilon = 1/64$. It is easy to check that $(\delta^2 / 4) m \ge n^{1-\epsilon}$ for large enough $n$ since $(\delta^2 / 4) m = n^{0.99} / 256$ and $n^{1-\epsilon} = n^{1-1/64} = o(n^{0.99})$. 
\end{proof}
\simpcollisionb*
\begin{proof}
    We apply \cref{lem:collision-for-d-1} with $\delta=4\log \log n / \log n$, $\epsilon = \log \log n / (16 \log n)$, and $m=n$ Then $(\delta^2 / 4) m = 4 n (\log \log n)^2 /  \log^2 n$, and $n^{1-\epsilon} = n \cdot \exp(-\log n \cdot \log \log n / (16 \log n)) = n \cdot \exp(-\log \log n / 16) = n \cdot \log^{-1/16} n$. Thus, the conditions of \cref{lem:collision-for-d-1} are satisfied.
\end{proof}

\section{Open questions}
\label{sec:open-questions}

\paragraph{Quantitative improvements.}
Our lower bounds can potentially be quantitatively improved in many ways: 
\begin{itemize}[noitemsep]
    \item Can the adaptive cell-probe bound be improved to, say $\Omega(\log n/\log \log n)$? 
    \item Can the distance bound be improved to $1-\exp(-n^{1-o(1)})$?
    \item What is the right bound for the number of nonadaptive probes in \cref{thm:main-nonadaptive}?
\end{itemize}
Our lower bound actually works for sampling $m=n^{1-\varepsilon}$ distinct elements that can be sampled with $O(\log n)$ adaptive \emph{bit-probes} \cite{Czumaj15}. Can one show a better upper bound for \emph{cell-probes} in that case, say $O(\log n / \log \log n)$?

Can any of the lower bounds be improved if $s$ is bounded? That would be sufficient for improving the data structure lower bounds in \cref{cor:ds-lowerbound}.

\paragraph{Other distributions.} 
  What \emph{symmetric} distributions in $\{0,1\}^n$ are samplable in $O(1)$ \emph{nonadaptive} cell-probes to $[n]^{\mathbb{N}}$? \cite{KOW25} show that essentially the only nontrivial distribution that can be sampled with \emph{nonadaptive bit-probes} is uniform over odd Hamming weight vectors. For the large input alphabet, this does not hold. For example, one can sample the uniform distribution over $\binom{[n]}{k}$ with $k$ nonadaptive cell-probes (as well as \emph{any} distribution uniform over a set of size $n^k$). On the other hand, it is hard even with $o(\log (n/k)/\log\log (n/k))$ adaptive \emph{bit} probes \cite{FLRS23}. 

  We conjecture that the uniform distribution over $\binom{[n]}{n/2}$ requires $\omega(1)$ adaptive cell-probes to sample. Showing this even in the nonadaptive case for a \emph{fixed constant} number of nonadaptive probes is open. 

  An intermediate challenge is to show that \emph{permutation matrices} are hard to sample with adaptive cell-probes. The target distribution is $\bm M \in \{0,1\}^{n \times n}$ and $\bm M$ is uniform over matrices with exactly one $1$-entry in every row and column. An efficient cell-probe sampler for $\bm M$ is not directly ruled out by our theorems even for the nonadaptive case, but we think that the technique should translate for that case as well.

\medskip
\subsection*{Acknowledgments}
We thank Ziyi Guan, Gilbert Maystre, and Weiqiang Yuan for discussions related to \cref{lem:coupling}.

\medskip

\DeclareUrlCommand{\Doi}{\urlstyle{sf}}
\renewcommand{\path}[1]{\small\Doi{#1}}
\renewcommand{\url}[1]{\href{#1}{\small\Doi{#1}}}
\bibliographystyle{alphaurl}
\bibliography{references}

@article{Combes15,
  author       = {Richard Combes},
  title        = {An extension of McDiarmid's inequality},
  journal      = {CoRR},
  volume       = {abs/1511.05240},
  year         = {2015},
  eprinttype    = {arXiv},
  eprint       = {1511.05240},
  bibsource    = {dblp computer science bibliography, https://dblp.org}
}

@online{Viola18,
  author       = {Viola, Emanuelle},
  title        = {The Complexity of Distributions},
  year         = {2018},
  url          = {https://www.youtube.com/live/O78b085HE3w?si=i7e44r9QuNzrR2dV&t=324},
  note         = {Talk at Simons Institute},
}

@article{Morris05,
	author = {Morris, Ben},
	doi = {10.1137/050636231},
	journal = {SIAM Journal on Computing},
	number = {2},
	pages = {484-504},
	title = {The Mixing Time of the Thorp Shuffle},
	volume = {38},
	year = {2008}}

@article{Bab87,
title = {Random oracles separate PSPACE from the polynomial-time hierarchy},
journal = {Information Processing Letters},
volume = {26},
number = {1},
pages = {51-53},
year = {1987},
issn = {0020-0190},
doi = {https://doi.org/10.1016/0020-0190(87)90036-6},
author = {Lászió Babai}
}

@book{Rud76,
  title={Principles of Mathematical Analysis},
  author={Rudin, W.},
  isbn={9780070856134},
  lccn={75179033},
  series={International series in pure and applied mathematics},
  year={1976},
  publisher={McGraw-Hill}
}

@misc{GKMOW25,
      title={Quantum Advantage from Sampling Shallow Circuits: Beyond Hardness of Marginals}, 
      author={Daniel Grier and Daniel M. Kane and Jackson Morris and Anthony Ostuni and Kewen Wu},
      year={2025},
      eprint={2510.07808},
      archivePrefix={arXiv},
      primaryClass={cs.CC}
}

@misc{BPP23,
      title={Unconditional Quantum Advantage for Sampling with Shallow Circuits}, 
      author={Adam Bene Watts and Natalie Parham},
      year={2023},
      eprint={2301.00995},
      archivePrefix={arXiv},
      primaryClass={quant-ph}
}

@article{Yao81, author = {Yao, Andrew Chi-Chih}, title = {Should Tables Be Sorted?}, year = {1981}, issue_date = {July 1981}, address = {New York, NY, USA}, volume = {28}, number = {3}, issn = {0004-5411}, doi = {10.1145/322261.322274}, journal = {J. ACM}, month = jul, pages = {615–628}, numpages = {14} }

@inbook{McDiarmid89, place={Cambridge}, series={London Mathematical Society Lecture Note Series}, title={On the method of bounded differences}, booktitle={Surveys in Combinatorics, 1989: Invited Papers at the Twelfth British Combinatorial Conference}, publisher={Cambridge University Press}, author={McDiarmid, Colin}, editor={Siemons, J.Editor}, year={1989}, pages={148–188}, collection={London Mathematical Society Lecture Note Series}}

@inproceedings{MRS09,
	address = {Berlin, Heidelberg},
	author = {Morris, Ben and Rogaway, Phillip and Stegers, Till},
	booktitle = {CRYPTO 2009},
	isbn = {978-3-642-03356-8},
	pages = {286--302},
	title = {How to Encipher Messages on a Small Domain},
	year = {2009}}

@BOOK{MT06,
  author={Montenegro, Ravi and Tetali, Prasad},
  title={Mathematical Aspects of Mixing Times in Markov Chains},
  year={2006},
  doi={10.1561/0400000003}}

@article{Morris09,
 author = {Ben Morris},
 journal = {The Annals of Probability},
 number = {2},
 pages = {453--477},
 title = {Improved Mixing Time Bounds for the Thorp Shuffle and L-Reversal Chain},
 urldate = {2025-10-02},
 volume = {37},
 year = {2009}
}

@article{FSS84,
	author = {Furst, Merrick and Saxe, James B. and Sipser, Michael},
	date = {1984/12/01},
	doi = {10.1007/BF01744431},
	id = {Furst1984},
	isbn = {1433-0490},
	journal = {Mathematical systems theory},
	number = {1},
	pages = {13--27},
	title = {Parity, circuits, and the polynomial-time hierarchy},
	volume = {17},
	year = {1984}}

@article{Markov1906,
author = {Markov, A.A.},
journal = {Journal Électronique d'Histoire des Probabilités et de la Statistique [electronic only]},
language = {eng},
number = {1b},
pages = {Article 10, 12 p., electronic only-Article 10, 12 p., electronic only},
title = {Extension of the law of large numbers to quantities, depending on each other (1906). Reprint.},
volume = {2},
year = {2006},
}

@article{Morris13,
author = {Morris, Ben},
title = {Improved mixing time bounds for the thorp shuffle},
year = {2013},
issue_date = {January 2013},
address = {USA},
volume = {22},
number = {1},
issn = {0963-5483},
doi = {10.1017/S0963548312000478},
journal = {Comb. Probab. Comput.},
month = jan,
pages = {118–132},
numpages = {15}
}

@InProceedings{reif1985,
  author       = {Reif, John H},
  booktitle    = {FOCS 1985},
  title        = {An optimal parallel algorithm for integer sorting},
  organization = {IEEE},
  pages        = {496--504},
  year         = {1985},
}

@article{Dur64,
author = {Durstenfeld, Richard},
title = {Algorithm 235: Random permutation},
year = {1964},
issue_date = {July 1964},
address = {New York, NY, USA},
volume = {7},
number = {7},
issn = {0001-0782},
doi = {10.1145/364520.364540},
journal = {Commun. ACM},
month = jul,
pages = {420},
numpages = {2}
}

@InProceedings{CGZ22,
  author =	{Chattopadhyay, Eshan and Goodman, Jesse and Zuckerman, David},
  title =	{{The Space Complexity of Sampling}},
  booktitle =	{ITCS 2022},
  pages =	{40:1--40:23},
  ISBN =	{978-3-95977-217-4},
  ISSN =	{1868-8969},
  year =	{2022},
  volume =	{215},
  address =	{Dagstuhl, Germany},
  URN =		{urn:nbn:de:0030-drops-156366},
  doi =		{10.4230/LIPIcs.ITCS.2022.40}
}

@inproceedings{FLRS23,
  author       = {Yuval Filmus and
                  Itai Leigh and
                  Artur Riazanov and
                  Dmitry Sokolov},
  
  title        = {Sampling and Certifying Symmetric Functions},
  booktitle    = {{APPROX/RANDOM} 2023},
  volume       = {275},
  pages        = {36:1--36:21},
  year         = {2023},
  doi          = {10.4230/LIPICS.APPROX/RANDOM.2023.36}
}

@article{CGFS86,
title = {Some intersection theorems for ordered sets and graphs},
journal = {Journal of Combinatorial Theory, Series A},
volume = {43},
number = {1},
pages = {23-37},
year = {1986},
issn = {0097-3165},
doi = {10.1016/0097-3165(86)90019-1},
author = {Fan Chung and Ronald Graham and Péter Frankl and James Shearer}
}

@Article{Viola20,
  author  = {Emanuele Viola},
  title   = {Sampling Lower Bounds: Boolean Average-Case and Permutations},
  doi     = {10.1137/18M1198405},
  number  = {1},
  pages   = {119--137},
  volume  = {49},
  journal = {{SIAM} Journal on Computing},
  year    = {2020},
}

@misc{KK18,
      title={First-Hitting Times Under Additive Drift}, 
      author={Timo Kötzing and Martin S. Krejca},
      year={2018},
      eprint={1805.09415},
      archivePrefix={arXiv},
      primaryClass={math.PR},
      url={https://arxiv.org/abs/1805.09415}, 
}

@inproceedings{KOW24,
  author       = {Daniel M. Kane and
                  Anthony Ostuni and
                  Kewen Wu},
  title        = {Locality Bounds for Sampling Hamming Slices},
  booktitle    = {STOC 2024},
  pages        = {1279--1286},
  year         = {2024},
  doi          = {10.1145/3618260.3649670},
}

@misc{KOW26,
  author       = {Daniel M. Kane and
                  Anthony Ostuni and
                  Kewen Wu},
  title        = {Symmetric Distributions from Shallow Circuits},
  booktitle    = {ECCC},
  year         = {2025},
  url          = {https://eccc.weizmann.ac.il/report/2025/183/},
}

@inproceedings{KOW25,
  author       = {Daniel M. Kane and
                  Anthony Ostuni and
                  Kewen Wu},
  title        = {Locally Sampleable Uniform Symmetric Distributions},
  booktitle    = {STOC 2025},
  pages        = {1807--1816},
  year         = {2025},
  doi          = {10.1145/3717823.3718243},
}

@Article{Viola12,
  author    = {Emanuele Viola},
  title     = {The Complexity of Distributions},
  doi       = {10.1137/100814998},
  number    = {1},
  pages     = {191--218},
  volume    = {41},
  bibsource = {dblp computer science bibliography, https://dblp.org},
  journal   = {{SIAM} Journal on Computing},
  year      = {2012},
}

@InProceedings{YZ24,
  author =	{Yu, Huacheng and Zhan, Wei},
  title =	{{Sampling, Flowers and Communication}},
  booktitle =	{ITCS 2024},
  pages =	{100:1--100:11},
  ISBN =	{978-3-95977-309-6},
  ISSN =	{1868-8969},
  year =	{2024},
  volume =	{287},
  doi =		{10.4230/LIPIcs.ITCS.2024.100},
}

@InProceedings{BIL12,
  author    = {Beck, Chris and Impagliazzo, Russell and Lovett, Shachar},
  booktitle = {FOCS 2012},
  title     = {Large Deviation Bounds for Decision Trees and Sampling Lower Bounds for AC0-Circuits},
  doi       = {10.1109/FOCS.2012.82},
  pages     = {101-110},
  year      = {2012},
}

@article {LV12,
    AUTHOR = {Lovett, Shachar and Viola, Emanuele},
     TITLE = {Bounded-depth circuits cannot sample good codes},
   JOURNAL = {Comput. Complexity},
  FJOURNAL = {Computational Complexity},
    VOLUME = {21},
      YEAR = {2012},
    NUMBER = {2},
     PAGES = {245--266},
      ISSN = {1016-3328},
   MRCLASS = {68Q17 (68P05 68Q87)},
  MRNUMBER = {2928214},
       DOI = {10.1007/s00037-012-0039-3},
}

@article{Vio12b,
	author = {Viola, Emanuele},
	journal = {SIAM Journal on Computing},
	number = {6},
	pages = {1593-1604},
	title = {Bit-Probe Lower Bounds for Succinct Data Structures},
	volume = {41},
	year = {2012}}

@InProceedings{CZ16,
  author    = {Chattopadhyay, Eshan and Zuckerman, David},
  booktitle = {STOC 2016},
  title     = {Explicit two-source extractors and resilient functions},
  doi       = {10.1145/2897518.2897528},
  isbn      = {9781450341325},
  location  = {Cambridge, MA, USA},
  pages     = {670–683},
  address   = {New York, NY, USA},
  numpages  = {14},
  year      = {2016},
}

@InProceedings{CS16,
  author    = {Cohen, Gil and Schulman, Leonard J.},
  booktitle = {FOCS 2016},
  title     = {Extractors for Near Logarithmic Min-Entropy},
  doi       = {10.1109/FOCS.2016.27},
  pages     = {178-187},
  year      = {2016},
}

@InProceedings{SS24,
  author    = {Shaltiel, Ronen and Silbak, Jad},
  booktitle = {STOC 2024},
  title     = {Explicit Codes for Poly-Size Circuits and Functions That Are Hard to Sample on Low Entropy Distributions},
  doi       = {10.1145/3618260.3649735},
  isbn      = {9798400703836},
  location  = {Vancouver, BC, Canada},
  pages     = {2028–2038},
  address   = {New York, NY, USA},
  numpages  = {11},
  year      = {2024},
}

@article{CK00, author = {Czumaj, Artur and Kutylowski, Miroslaw}, title = {Delayed path coupling and generating random permutations}, year = {2000}, issue_date = {Oct. 2000}, address = {USA}, volume = {17}, number = {3–4}, issn = {1042-9832}, journal = {Random Struct. Algorithms}, month = oct, pages = {238–259}, numpages = {22} }

@InProceedings{Yu2020,
  author    = {Yu, Huacheng},
  booktitle = {STOC 2020},
  title     = {Nearly optimal static Las Vegas succinct dictionary},
  doi       = {10.1145/3357713.3384274},
  isbn      = {9781450369794},
  location  = {Chicago, IL, USA},
  pages     = {1389–1401},
  address   = {New York, NY, USA},
  numpages  = {13},
  year      = {2020},
}

@article{MRRR12, author = {Munro, Ian and Raman, Rajeev and Raman, Venkatesh and Srinivasa, Rao}, title = {Succinct representations of permutations and functions}, year = {2012}, issue_date = {June, 2012}, address = {GBR}, volume = {438}, issn = {0304-3975},  doi = {10.1016/j.tcs.2012.03.005},  journal = {Theor. Comput. Sci.}, month = jun, pages = {74–88}, numpages = {15} }

@InProceedings{Golynsky09,
  author    = {Alexander Golynski},
  booktitle = {SODA 2009},
  title     = {Cell Probe Lower Bounds For Succinct Data Structures},
  doi       = {10.1137/1.9781611973068.69},
  pages     = {625-634},
  year      = {2009},
}

@inproceedings{Hag91,
	address = {Berlin, Heidelberg},
	author = {Hagerup, Torben},
	booktitle = {Automata, Languages and Programming},
	pages = {405--416},
	title = {Fast parallel generation of random permutations},
	year = {1991}}

@InProceedings{MV91,
  author    = {Matias, Yossi and Vishkin, Uzi},
  booktitle = {STOC 1991},
  title     = {Converting high probability into nearly-constant time—with applications to parallel hashing},
  doi       = {10.1145/103418.103453},
  isbn      = {0897913973},
  location  = {New Orleans, Louisiana, USA},
  pages     = {307–316},
  address   = {New York, NY, USA},
  numpages  = {10},
  year      = {1991},
}

@article{Viola14,
	author = {Viola, Emanuele},
	journal = {SIAM Journal on Computing},
	number = {2},
	pages = {655-672},
	title = {Extractors for Circuit Sources},
	volume = {43},
	year = {2014}}

@InProceedings{Czumaj15,
  author    = {Artur Czumaj},
  booktitle = {STOC 2015},
  title     = {Random Permutations using Switching Networks},
  doi       = {10.1145/2746539.2746629},
  pages     = {703--712},
  year      = {2015},
}

@inproceedings{Viola23,
  author       = {Emanuele Viola},
  title        = {New Sampling Lower Bounds via the Separator},
  booktitle    = {{CCC} 2023},
  volume       = {264},
  pages        = {26:1--26:23},
  doi          = {10.4230/LIPICS.CCC.2023.26},
  year         = {2023}
}

@article{Harper1966,
  author  = {Harper, L. H.},
  title   = {Optimal Numberings and Isoperimetric Problems on Graphs},
  journal = {Journal of Combinatorial Theory},
  volume  = {1},
  pages   = {385--393},
  year    = {1966}
}

@book{InfTheoryBook,
author = {Cover, Thomas M. and Thomas, Joy A.},
title = {Elements of Information Theory (Wiley Series in Telecommunications and Signal Processing)},
year = {2006},
isbn = {0471241954},
publisher = {Wiley-Interscience},
address = {USA}
}

@article{Thorp73,
 author = {Edward Thorp},
 journal = {Journal of the American Statistical Association},
 number = {344},
 pages = {842--847},
 title = {Nonrandom Shuffling with Applications to the Game of Faro},
 urldate = {2025-11-02},
 volume = {68},
 year = {1973},
doi ={10.2307/2284510}
}

\end{document}